\documentclass[aps.pra,onecolumn,showpacs,preprintnumbers,amsmath,amssymb]{revtex4-2}
    
    \usepackage[utf8]{inputenc}
    \usepackage[english]{babel}
    \newtheorem{theorem}{Theorem}
    \newtheorem{corollary}{Corollary}
    \newtheorem{example}{Example}
    \newtheorem{lemma}{Lemma}
    \usepackage{float}
    \newtheorem{definition}{Definition}
    \newtheorem{remark}{Remark}
    \usepackage{setspace}
    \usepackage{booktabs}
    \usepackage{multirow}
    \usepackage{amsmath,amssymb}
    \usepackage{qcircuit}
    \usepackage{caption}
    \usepackage{lipsum}
    \usepackage{tabularx}
    \usepackage{natbib}
    \usepackage{stfloats}
    \usepackage{graphicx}
    \usepackage{rotating}
    \usepackage{subcaption}

    \usepackage{makecell}
    \setcellgapes{2.3pt}
    
    \usepackage{subcaption} 
    \newenvironment{proof}{\textit{Proof.}}{\hfill$\square$}
    \usepackage{physics}
    \usepackage[figurename=Fig., justification=RaggedRight]{caption}
    \usepackage{amsmath}
    \usepackage{xcolor}
    \usepackage{hyperref}

    \usepackage{adjustbox,expl3,etoolbox}
    \letcs\replicate{prg_replicate:nn}
    
\begin{document} 
    
    \title{Quantum Framework for Simulating Linear PDEs with Robin Boundary Conditions}
\author{Nikita Guseynov}
\email{guseynov.nm@gmail.com}
\affiliation{University of Michigan-Shanghai Jiao Tong University Joint Institute, Shanghai 200240, China.}

\author{Xiajie Huang}
\email{xj.huang@sjtu.edu.cn}
\affiliation{School of Mathematical Sciences, Shanghai Jiao Tong University, Shanghai 200240, China}
\affiliation{Shanghai Artificial Intelligence Laboratory, Shanghai, China}

\author{Nana Liu}
\email{nana.liu@quantumlah.org}
\affiliation{Institute of Natural Sciences, School of Mathematical Sciences, Shanghai Jiao Tong University, Shanghai 200240, China}
\affiliation{Ministry of Education Key Laboratory in Scientific and Engineering Computing, Shanghai Jiao Tong University, Shanghai 200240, China}
\affiliation{Shanghai Artificial Intelligence Laboratory, Shanghai, China}
\affiliation{University of Michigan-Shanghai Jiao Tong University Joint Institute, Shanghai 200240, China.}

\begin{abstract}
We propose an explicit, oracle-free quantum framework for numerically simulating general linear partial differential equations (PDEs), extending previous work \cite{guseynov2024efficientPDE} to incorporate (a) Robin boundary conditions—which include Neumann and Dirichlet conditions as special cases—(b) inhomogeneous terms, and (c) variable coefficients in space and time. Our approach begins with a general finite-difference discretization and applies the Schr\"{o}dingerisation technique to transform the resulting system into one that admits unitary quantum evolution, enabling quantum simulation.

For the Schr\"{o}dinger equation corresponding to the discretized PDE, we construct an efficient block-encoding of the Hamiltonian $H$ that scales polylogarithmically with the number of grid points $N$. This encoding is compatible with quantum signal processing and allows for the implementation of the evolution operator $e^{-iHt}$. The oracle-free nature of our method permits complexity to be measured in fundamental gate units—namely, CNOT gates and single-qubit rotations—bypassing the inefficiencies of oracle queries. Consequently, the overall algorithm scales polynomially with $N$ and linearly with the spatial dimension $d$, achieving a polynomial speedup in $N$ and an exponential advantage in $d$, thereby mitigating the classical curse of dimensionality. The validity and efficiency of the proposed approach are further substantiated by numerical simulations.

By explicitly defining the quantum operations and quantifying their resource requirements, our approach offers a practical alternative for numerically solving PDEs, distinct from others that rely on oracle queries and purely asymptotic scaling methods.
\end{abstract}

    \maketitle
    
    \section{Introduction}\label{sec: intro}
Partial differential equations (PDEs) provide a unifying mathematical  
description of a wide range of natural and engineered systems.  
However, their solution often demands significant computational  
resources, especially in high-dimensional or finely discretized  
settings \cite{evans2022partial}. Quantum computing \cite{FEYNMAN} has the potential to solve certain problems more efficiently than classical computers. Recently, it has attracted growing interest for accelerating computations involving partial differential equations (PDEs) \cite{laughlin2000theory}, with  
applications emerging in physics \cite{costa2019quantum,gaitan2020finding},  
engineering \cite{jin2023time,linden2022quantum,sato2024hamiltonian},  
and finance \cite{stamatopoulos2020option,gonzalez2023efficient}.  

These developments highlight a more general issue in  
quantum algorithm design: how to account for the cost of  
representing problem structure. Many quantum algorithms  
rely heavily on quantum oracles or block-encoding techniques to encode the structure of the computational problem. While these approaches estimate complexity in terms of queries to  
the oracle or block-encoding, such measures may fail to accurately capture  
the practical computational advantage. In reality, the runtime of these algorithms  
may be dominated by the inefficiencies involved in constructing these oracles  
and block-encodings. Algorithms such as Grover’s search \cite{grover}, quantum phase estimation  
\cite{nielsen2002quantum}, the quantum linear systems algorithm (HHL) \cite{hhl}, quantum random walks  
\cite{berry2015hamiltonian}, and thermal state preparation \cite{rall2023thermal}  
rely on these constructions. Thus, assessing quantum algorithm efficiency requires analyzing the actual  
runtime, not just oracle or query complexity. This avoids hidden  
costs arising from inefficient oracle or block-encoding constructions and provides  
a more realistic measure of performance.  

In this work, we focus on developing explicit quantum schemes for a fully fault-tolerant quantum  
computing setup. Although the T gate plays a foundational  
role in such architectures, its practical use can be subtle.  
Even a single-qubit \( R_z \) rotation may require a variable  
number of T gates depending on its parameters, particularly  
when approximated via the Solovay–Kitaev algorithm \cite{dawson2005solovay}.  
For this reason, we treat CNOT gates and single-qubit rotations as  
elementary building blocks, due to their consistency and  
predictability in gate-count analysis.  

This work builds upon our previous study \cite{guseynov2024efficientPDE}, where we  
focused on simulating PDEs with periodic boundary  
conditions—considered as the fundamental and simplest case. In that work, we introduced  
explicit quantum circuits for simulating PDEs using basic quantum gates, avoiding the need  
for quantum oracles or block-encodings. Here, we generalize this approach to handle more  
general boundary conditions and discretization schemes, while preserving the oracle-free setup.  
In our method, we assume that the initial condition of the PDE at time \( t = 0 \) is given as a quantum  
state, and then construct a quantum circuit that transforms this initial state into  
the final state at time \( t \). We compare this circuit's complexity with that of classical algorithms  
for solving PDEs and demonstrate certain advantages, including a polynomial speedup in  
spatial resolution and an exponential speedup in the number of spatial dimensions.  

To solve a discretized linear PDE, we first apply the  
Schr\"{o}dingerisation technique \cite{PRL2024, PRA2023,analog,cao2023quantum,PRS2024} to transform the matrix form of the  
PDE into a system evolving via a unitary transformation (with a corresponding Hermitian Hamiltonian) in an extended Hilbert space.  
If the coefficients in the PDE depend on time (i.e., non-autonomous PDEs), we can still use quantum simulation for time-independent Hamiltonians by introducing an additional spatial dimension to act as a clock dimension \cite{cao2023quantum, cao2024unifying}.  
After these two steps, the resulting Hamiltonian becomes time-independent.  
The key component of our method is the  
block-encoding of this time-independent Hamiltonian \( H \).  
This block-encoding is then used as input for quantum Hamiltonian simulation via quantum singular value transformation (QSVT) \cite{gilyen2019quantum}, enabling implementation of the evolution operator \( e^{-iHt} \).  
Applying \( e^{-iHt} \) to the encoded initial state and  
postselecting on auxiliary registers (a simple procedure)  
yields the quantum state corresponding to \( u(t, \vec{x}) \).  

This paper is organized as follows. Sec.~\ref{section: 
result and related work} presents a general overview of the PDE (\ref{eq: PDE 1D MAIN}), provides  
the complexity of our quantum numerical method, and compares it  
with state-of-the-art quantum and classical methods.  
Sec.~\ref{section: 
discretization schrodingerisation and time dependence}  
explains how the problem can be reduced to a  
Hamiltonian simulation task with time-independent \( H \),  
using the Schr\"{o}dingerisation technique. This technique transforms a non-conservative PDE with a non-Hermitian operator into a Schr\"{o}dinger equation in an extended  
Hilbert space of dimension \( d + 1 \). In the case of time-dependent \( H \), we further extend the Hilbert space  
by one clock dimension \cite{cao2023quantum}, introducing an  
auxiliary variable that behaves like time within the Hamiltonian.  
The quantum state is localized around \( s = t \), effectively  
recovering time-dependent dynamics within a static framework.  
Sec.~\ref{sec: 1D} provides a step-by-step construction of the block-encoding of  
the Hamiltonian in the 1D case.  
Sec.~\ref{sec: multidim} generalizes the approach to  
multidimensional PDEs, constructing block-encodings  
for higher dimensions within a unified framework.  
The block-encoding is then used for quantum Hamiltonian  
simulation using QSVT in Sec.~\ref{section: QSVT}.  
Finally, Sec.~\ref{section: numericals} presents numerical  
experiments, and Sec.~\ref{section: conclusion} concludes the paper and outlines future directions.

\section{Main result and related work}\label{section: result and related work}

We begin with a one-dimensional linear partial differential equation  
(PDE) as in Eq.~(\ref{eq: PDE 1D MAIN}), which serves as a fundamental case for building quantum algorithms  
for numerical simulation. This formulation is widely used to model physical processes such as diffusion, transport, and wave propagation. Despite its simplicity, the 1D case includes key features like variable coefficients, higher-order spatial derivatives, and general boundary conditions, making it  
suitable for developing and testing core ideas.

\begin{eqnarray}
\begin{gathered}
    \frac{\partial u(x,t)}{\partial t}=\hat{A}u(x,t)+v(x,t);\qquad\hat{A}=\sum_{k=0}^{\eta-1}f_k(x,t)\frac{\partial^{p_k}}{\partial x^{p_k}};\\
    u(x,t=0)=u_0(x);\qquad \frac{\partial u(x,t)}{\partial x}\Bigg|_{x=a}+\mathcal{A}_1 u(x=a,t)=\mathcal{A}_2;\qquad \frac{\partial u(x,t)}{\partial x}\Bigg|_{x=b}+\mathcal{B}_1 u(x=b,t)=\mathcal{B}_2.
\end{gathered}
    \label{eq: PDE 1D MAIN}
\end{eqnarray}
In Eq.~(\ref{eq: PDE 1D MAIN}), the coefficients $\mathcal{A}_1$, $\mathcal{A}_2$, $\mathcal{B}_1$, $\mathcal{B}_2$ are complex-valued constants,  
and each function $f_k(x,t)$ is assumed to be piece-wise  
continuous on the domain $x \in [a, b]$. The source function  
$v(x,t)$ is also taken to be piece-wise continuous. The PDE  
is defined on a finite interval with Robin-type boundary  
conditions, which incorporate both function values and their  
spatial derivatives at $x = a$ and $x = b$. Although Dirichlet  
and Neumann boundary conditions are not directly written in  
the form of Eq.~(\ref{eq: PDE 1D MAIN}), both cases are naturally included  
within the Robin framework and can be recovered through suitable modifications to the  
discretized system and corresponding quantum circuit. In addition, it's straightforward to consider higher-order derivatives in time \cite{analog}.

As a key result of this work, we present Thm.~\ref{theorem: intro 1d}, which summarizes the quantum numerical approach to solve the one-dimensional PDE.

\begin{theorem}[One-dimensional PDE (informal version of Thm.~\ref{theorem: evolution 1 D})]\label{theorem: intro 1d}
Let \(T \in (0, \infty)\) denote the evolution time, and let \(\epsilon \in (0, 1)\) represent the allowable error. Consider the PDE given by Eq.~(\ref{eq: PDE 1D MAIN}), where the spatial dimension \(x\) is discretized into \(N = 2^n\) points, with \(n\) the number of qubits. Suppose we have an \(n\)-qubit quantum state \(\ket{u_0}^n\) encoding the discretized initial condition vector \(\vec{u}(t=0)\).

Then, there exists a quantum gate sequence that, with success probability 
\[
p_{\mathrm{success}} \sim \frac{\|\vec{u}(t=T)\|^2}{\|\vec{u}(t=0)\|^2},
\]
transforms the initial state \(\ket{u_0}^n\) into a state proportional to the solution vector at time \(T\), i.e., \(\vec{u}(t=0) \rightarrow \vec{u}(t=T)\).

The total resource cost scales as:
\begin{enumerate}
    \item \(\displaystyle
    \mathcal{O}\Bigg(
    \underbrace{
    \left[
    \eta \kappa n + GQ \eta n \log n
    \right]
    }_{\text{block-encoding of } H}\cdot\underbrace{
    \left[
     T\underbrace{\kappa \sum_{k=0}^{\eta - 1} 2^{np_k} \mathcal{N}_{f_k}}_{||H||_{\max}}
    + \frac{\ln(1/\epsilon)}{\ln\biggl(e + \frac{\ln(1/\epsilon)}{\kappa T \sum_{k=0}^{\eta - 1} 2^{np_k} \mathcal{N}_{f_k}}\biggr)}
    \right]
    }_{\text{QSVT: }H\rightarrow e^{-iHT}}
    \Bigg)
    \)
    quantum gates,
    \item \(\mathcal{O}(n)\) additional qubits,
\end{enumerate}
where \( ||H||_{\max} \) denotes the maximum absolute eigenvalue of
the Hamiltonian, as determined by the Quantum Singular Value Transformation
(QSVT), see Thm.~\ref{theorem: optimal block-hamiltonian simultaions} and
Remark~\ref{remark: H_max}. It has been proven that a linear scaling
in terms of \( ||H||_{\max} \) is optimal \cite{low2017optimal}. The
coefficient \( N_{f_k} \sim \max_{x \in [a,b]} |f_k(x)| \) represents the
normalization constant associated with each function \( f_k \), while \( G \)
indicates the maximum number of pieces among all piecewise continuous functions
\( f_k \) and \( v \). Additionally, \( Q \) refers to the maximum polynomial
degree within a single segment of the decompositions of \( f_k \) and \( v \).
The parameter \( \kappa \) characterizes the sparsity of the discretized operator
\( \hat{A} \), which is typically proportional to the highest order of the
spatial derivative \( \max_k p_k \).

\end{theorem}

 While we first focus on the 1D case as a fundamental one, the overall framework  
extends naturally to multi-dimensional domains. We present the generalized $d$-dimensional form

\begin{eqnarray}
    \frac{\partial u(x_1,x_2,\dots,x_d,t)}{\partial t} =
\hat{A}^{(d)} u(x_1,\dots,x_d,t) + v^{(d)}(x_1,\dots,x_d,t); \quad
\hat{A}^{(d)} = \sum_{k=0}^{\eta-1} f_k^{(d)}(x_1,\dots,x_d,t)
\frac{\partial^{p^{(1)}_k}}{\partial x_1^{p^{(1)}_k}} \cdots
\frac{\partial^{p^{(d)}_k}}{\partial x_d^{p^{(d)}_k}};
\label{eq: intro PDE multidim}
\end{eqnarray}

In Sec.~\ref{sec: multidim} we consider how to extend the block-encoding construction to this important case; in the Sec.~\ref{section: QSVT} we further explain how using QSVT solve the PDE (\ref{eq: intro PDE multidim}) with resources (informal view) scales as
\begin{enumerate}
    \item \(\displaystyle
    \mathcal{O}\Bigg(
    \underbrace{
    \left[
    d\eta M Q_{\text{PET}} G Q n \log n+d \eta \kappa n
    \right]
    }_{\text{block-encoding of } H}\cdot\underbrace{
    \left[
    T \underbrace{\kappa \sum_{k=0}^{\eta - 1} 2^{n\tilde{p}_k} \mathcal{N}_{f_k}}_{||H||_{\max}}
    + \frac{\ln(1/\epsilon)}{\ln\biggl(e + \frac{\ln(1/\epsilon)}{\kappa T \sum_{k=0}^{\eta - 1} 2^{n\tilde{p}_k} \mathcal{N}_{f_k}}\biggr)}
    \right]
    }_{\text{QSVT: }H\rightarrow e^{-iHT}}
    \Bigg)
    \)
    quantum gates,
    \item $\mathcal{O}(n)$ additional qubits,
\end{enumerate}
where $\tilde{p}_k = \sum_i p_k^{(i)}$ is the accumulated derivative  
degree across all dimensions; the full $MQ_{\text{PET}}GQ$ factor quantifies the complexity
of functions $f_k^{(d)}$ and $v^{(d)}$ appearing in the PDE  
In this work, we represent each function as a product of univariate  
functions, each given by a degree-$Q$ polynomial with $G$ segments  
We then construct $M$ such factorized components and sum them linearly (LCU).  
On top of this superposition, we apply a polynomial transformation  
of degree $Q_{\text{PET}}$ (QSVT), analogous to activation in neural networks  
This procedure can be iterated, yielding $(MQ_{\text{PET}})^2$, $(MQ_{\text{PET}})^3$, etc.  
However, we restrict to one transformation layer which suffices for many physically relevant functions; see Table~\ref{table: amplitude-oracle-functions}  
and Sec.~\ref{sec: multidim}.

We first compare our algorithm to classical finite-difference approaches, including the forward Euler method, and implicit schemes such as Crank–Nicolson and the backward Euler method. We count floating-point operations (FLOPs) as the complexity unit.
These methods discretize time and $d$-dimensional space to iteratively evolve
the solution via matrix updates.
Let $N = 2^n$ denote the number of grid points per dimension, so the total discretized
domain contains $N^d$ points.
We next address the time discretization $\Delta t$; it is well known \cite{leveque2007finite}
that the total error for numerical PDE solution is
\begin{equation}
    \epsilon \sim \mathcal{O}(\Delta t) + \mathcal{O}(\Delta x^g),
\end{equation}
where $g$ is the maximal cumulative accuracy order of the finite-difference derivative
approximation, see Table~\ref{table: famous schemes}.
Thus, to ensure good error scaling, the time step must satisfy
$\Delta t \sim (\Delta x)^g$, so simulating total time $T$ requires
$T/\Delta t \sim T N^g$ time steps.
Each such step involves $\mathcal{O}(2^n)$ or $\mathcal{O}(N)$ FLOPs
across all $d$ dimensions.
Therefore, the total classical complexity scales as
$\mathcal{O}(T 2^{g n + n d})$ or equivalently $\mathcal{O}(T N^{g+d})$ for simulating
the PDE~(\ref{eq: intro PDE multidim}), see more details in the Appendix~\ref{appendix: complexity analysis for classical method}.
In contrast, our method exhibits complexity
$\mathcal{O}(d T N^m \log N \log\log N)$ or
$\mathcal{O}(d T 2^{n m} n \log n)$ under certain simplifications,
where $m$ is the maximal cumulative derivative order among PDE terms.
For example, $m = 2$ for the term $\partial^2 u/\partial x \partial y$. 
The most notable advantage is the exponential speedup in $d$,
as the classical cost grows as $N^d$, while our approach scales
linearly with $d$. We also note that the exponents $g$ (classical)
and $m$ (quantum) reflect different aspects: $g$ is set by the
finite-difference accuracy, while $m$ is the maximal derivative order
in the PDE. For a wide class of problems, $g \sim m$ or even $g > m$,
so our method offers not only exponential advantage in $d$, but also
a polynomial advantage in $N$ over classical algorithms for a wide class of PDEs.

It is important to note that this scaling accounts only for Hamiltonian evolution and omits the costs of
state preparation and measurement, which are context-dependent. The quantum output is a quantum state
approximating the solution, requiring additional processing for interpretation
in practice. While this scaling gap reveals a promising quantum advantage, it must be interpreted with caution, as inefficient state preparation or measurement could offset the potential benefits of the quantum algorithm in practice. In addition,
our analysis assumes that the input functions $f_k^{(d)}$ and $v^{(d)}$
admit the structured form described above; significant deviations from
this model could also diminish performance.

Among quantum methods, the work~\cite{kharazi2025explicit} suggests a particularly
promising way to create block-encodings for finite-difference operators, which
can then be used for Hamiltonian simulation to obtain numerical solutions of PDEs.
In contrast to our approach, their main difference lies in the construction of block-encodings for derivative operators: the authors employ permutation operators
as the basic unit for building banded sparse matrices, while we utilize
the Banded-sparse-access construction. For periodic boundary conditions, both
methods achieve the same asymptotic scaling; moreover, we acknowledge that
the method in~\cite{kharazi2025explicit} requires fewer pure ancilla qubits.
However, for more general boundaries such as Dirichlet, Neumann, or Robin,
which are the main target of this paper, their method exhibits higher complexity,
as it does not provide efficient element-wise access to matrix entries as illustrated
in the central part of Fig.~\ref{fig:1 term ROBIN} and Fig.~\ref{fig: derivative circuit}.
The authors briefly discuss how Dirichlet and other boundaries can be treated
using the LCU technique, but their approach is less efficient for such cases.

Many efficient quantum simulation algorithms for Hamiltonian simulation,
such as quantum random walk with quantum signal processing~\cite{low2017optimal}
and continuous or methods that using fractional query models~\cite{berry2014exponential}, assume
access to a digital quantum oracle, that is, a unitary operator of the form
\begin{equation}
    \mathcal{O}_H\,|i\rangle|j\rangle|0\rangle = |i\rangle|j\rangle|H_{ij}\rangle,
\end{equation}
which returns Hamiltonian matrix elements in binary encoding. While those algorithms
demonstrate optimal dependence on all parameters, their practical application is
hindered by the need to construct these oracles using quantum arithmetic circuits,
whose gate complexity and ancilla requirements are known to be substantial~\cite{munoz2018t,haner2018optimizing}.
This makes implementation challenging, especially for general or piece-wise continuous
functions and for problems with nontrivial boundary conditions. In comparison, we
provide explicit and efficient constructions of amplitude block-encodings for Hamiltonians
with piece-wise continuous coefficients and general boundaries. These amplitude oracles
are then employed in specific Hamiltonian simulation techniques to solve PDEs.

Other literature~\cite{sato2024hamiltonian,hu2024quantum} explores the construction of quantum
circuits for PDE-induced Hamiltonians $\hat{H}$ via decomposition into $2\times 2$ ladder
operators, serving as an alternative to the well-established Pauli basis~\cite{nielsen2002quantum}.
While promising for models with separate momentum and coordinate operators, these approaches
require further development to efficiently address more general cases, such as $xp$ terms,
piecewise $x$-dependent coefficients in PDEs, and sophisticated boundary conditions.

It is worth mentioning that one of the earliest quantum algorithms
for PDEs was proposed by Cao et al.~\cite{cao2013quantum}, who designed
a quantum circuit for solving the Poisson equation. Their method is
based on the HHL quantum linear systems algorithm, with explicit circuit
modules for Hamiltonian simulation and eigenvalue inversion, achieving
exponential speedup in the dimension $d$. The approach features nearly
linear scaling in $d$ and polylogarithmic scaling in the target accuracy,
but is limited to regular grids, constant coefficients, and Dirichlet
boundary conditions. Further development is needed to address PDEs with
variable or piece-wise continuous coefficients and more general boundaries.

    \section{Quantum approach for numerical simulating of linear PDEs}\label{section: discretization schrodingerisation and time dependence}
    
In this section, we demonstrate how the numerical solution of a general
class of linear partial differential equations (PDEs), as described by
Eq.~\ref{eq: PDE 1D MAIN}, can be reformulated as a Hamiltonian simulation
problem. This transformation results in a Schrödinger equation governed by a
time-independent Hamiltonian \(H\). The key steps involved are: (i) applying
the finite-difference method to discretize the PDE into matrix form; (ii)
introducing an auxiliary qubit to convert the PDE into a homogeneous form;
(iii) utilizing the Schrödingerisation technique to obtain a PDE characterized
by a Hermitian matrix \(A\); and (iv) extending the Hilbert space by adding
an auxiliary "clock" dimension, which converts the time dependence of \(H\)
into a spatial coordinate in the new dimension. The resulting time-independent
Schrödinger equation is subsequently solved as detailed in Sec.~\ref{sec: 1D},
\ref{sec: multidim}, and \ref{section: QSVT}.

As a first step, we discretize the spatial derivatives by employing finite-difference schemes with adjustable accuracy. These include backward, forward,
and central methods. Each scheme offers distinct trade-offs between
computational cost and approximation precision, as summarized in
Table~\ref{table: famous schemes}.

\renewcommand{\arraystretch}{1.6}
\begin{table}[H]
\centering
\begin{tabular}{|c|c|c|>{\raggedright\arraybackslash}p{8.2cm}|}
\hline
\textbf{Type} & \textbf{Der. order} & \textbf{Acc. order} & \multicolumn{1}{c|}{\textbf{Formula}} \\
\hline
Backward & First & First & \hspace{1cm}$\frac{\partial u}{\partial x} \Big|_{x=x_i} \approx \frac{u_i - u_{i-1}}{\Delta x} + \mathcal{O}(\Delta x)$ \\
\hline
Forward & First & Second & \hspace{1cm}$\frac{\partial u}{\partial x} \Big|_{x=x_i} \approx \frac{-u_{i+2} + 4u_{i+1} - 3u_i}{2\Delta x} + \mathcal{O}((\Delta x)^2)$ \\
\hline
Central & First & Second & \hspace{1cm}$\frac{\partial u}{\partial x} \Big|_{x=x_i} \approx \frac{u_{i+1} - u_{i-1}}{2\Delta x} + \mathcal{O}((\Delta x)^2)$ \\
\hline
Backward & Second & First & \hspace{1cm}$\frac{\partial^2 u}{\partial x^2} \Big|_{x=x_i} \approx \frac{u_i - 2u_{i-1} + u_{i-2}}{(\Delta x)^2} + \mathcal{O}(\Delta x)$ \\
\hline
Central & Second & Second & \hspace{1cm}$\frac{\partial^2 u}{\partial x^2} \Big|_{x=x_i} \approx \frac{u_{i-1} - 2u_i + u_{i+1}}{(\Delta x)^2} + \mathcal{O}((\Delta x)^2)$ \\
\hline
Forward & Second & Second & \hspace{1cm}$\frac{\partial^2 u}{\partial x^2} \Big|_{x=x_i} \approx \frac{2u_i - 5u_{i+1} + 4u_{i+2} - u_{i+3}}{(\Delta x)^2} + \mathcal{O}((\Delta x)^2)$ \\
\hline
\end{tabular}
\caption{Examples of finite-difference schemes for approximating first and second
spatial derivatives of $u(x)$ at grid point $x_i$, categorized by type,
derivative order, and accuracy order in terms of \(\Delta x\).}\label{table: famous schemes}
\end{table}

\begin{eqnarray}
\begin{gathered}
    \frac{\partial u}{ \partial x}\Bigg|_{x=x_i} \approx\frac{1}{\Delta x}\sum_{m=i-s^L_1}^{i+s^R_1}\gamma_{1m} u_m+\mathcal{O}((\Delta x)^g );\\
    \frac{\partial^2 u}{ \partial x^2}\Bigg|_{x=x_i} \approx \frac{1}{(\Delta x)^2}\sum_{m=i-s^L_2}^{i+s^R_2}\gamma_{2m} u_m+\mathcal{O}((\Delta x)^g ),\\
    \vdots
\end{gathered}
\end{eqnarray}
where $g$ denotes the order of accuracy of the scheme.. We now discretize the full PDE system as follows:
\begin{eqnarray}
\begin{gathered}
    \frac{\partial \vec{u}}{\partial t} = A\vec{u} + \vec{v}; \quad  \vec{u} = \left( u_0 \;\; u_1 \;\; \cdots \;\; u_{N-1} \right)^\top;\qquad \vec{v}=\left(v_0+v^\prime_0,v_1+v^\prime_1,\dots,v_{N-1}+v^\prime_{N-1}\right)^\top;\\
    A=\sum_kA_k;\qquad \left(A_k\vec{u}\right)_i= \frac{f_k(x_i,t)}{(\Delta x)^{p_k}}\sum_{m=i-s^L_{p_k}}^{i+s^R_{p_k}}\gamma_{p_km} u_m\approx \left. f_k(x,t)\frac{\partial^{p_k}u(x,t)}{\partial x^{p_k}} \right|_{x = x_i},
    \end{gathered}
    \label{eq:PDE discretized}
\end{eqnarray}
The matrix \( A \) is a \( 2^n \times 2^n \) discretized version of the operator \( \hat{A} \), as described in Eq.~(\ref{eq: PDE 1D MAIN}). \( A_k \) represents the discretized version of the \( k \)-th term of \( A \). The notation \( \left(A_k\vec{u}\right)_i \) refers to the \( i \)-th matrix element of the resulting vector \( A_k \vec{u} \). The corrections \( v_i' \) account for enforcing boundary conditions and are typically nonzero near the domain boundaries, as elaborated further in Example~\ref{example: discretization}.

\begin{example}\label{example: discretization}
As an illustration, we consider the spatial discretization of the
one-dimensional heat equation with a general Robin boundary condition:
\begin{eqnarray}
\begin{gathered}
    \frac{\partial u(x,t)}{\partial t} = \frac{\partial^2 u(x,t)}{\partial x^2}, \\
    \frac{\partial u}{\partial x}\Big|_{x=a} + \mathcal{A}_1 u(a,t) = \mathcal{A}_2, \qquad 
    \frac{\partial u}{\partial x}\Big|_{x=b} + \mathcal{B}_1 u(b,t) = \mathcal{B}_2.
\end{gathered}
\end{eqnarray}

The interval $[a,b]$ is uniformly discretized with step size $\Delta x$, generating grid points $x_i = a + i\Delta x$ for $i = 0, \dots, N-1$. A fourth-order central finite difference scheme is used to approximate 
$\partial^2 u / \partial x^2$, requiring values from $u_{i-2}$ to $u_{i+2}$. 
Near boundaries, ghost points are eliminated using the boundary 
conditions, modifying the final few rows.

The resulting approximations are:
\begin{eqnarray}
\begin{gathered}
\frac{\partial^2 u}{\partial x^2} \Big|_{x_i} \approx \frac{-u_{i-2}/12 + 4u_{i-1}/3 + -5u_{i}/2 + 4u_{i+1}/3-u_{i+2}/12}{\Delta x^2}, \quad i = 2, \dots, N-3; \\
\frac{\partial^2 u}{\partial x^2} \Big|_{x_0} \approx \frac{(7\mathcal{A}_1\Delta x/3 -5/2)u_{0} + 8u_{1}/3-u_{2}/6}{\Delta x^2}\underbrace{-\frac{7\mathcal{A}_2}{3\Delta x}}_{v^\prime(x_0)}\\
\frac{\partial^2 u}{\partial x^2} \Big|_{x_1} \approx \frac{ \left(4 / 3-\mathcal{A}_1\Delta x / 6 \right)u_0  - 31u_1/12 + 4u_2 / 3 - u_3 / 12 }{\Delta x^2} +\underbrace{\frac{\mathcal{A}_2}{6\Delta x}}_{v^\prime (x_1)};\\
\frac{\partial^2 u}{\partial x^2} \Big|_{x_{N-2}} \approx \frac{-u_{N-4}/12 + 4u_{N-3}/3 + -31u_{N-2}/12+(4/3+\mathcal{B}_1\Delta x/6)u_{N-1}}{\Delta x^2}\underbrace{-\frac{\mathcal{B}_2}{6\Delta x}}_{v^\prime(x_{N-2})};\\
\frac{\partial^2 u}{\partial x^2} \Big|_{x_{N-1}} \approx \frac{-u_{N-3}/6 + 8u_{N-2}/3 + (-5/2-7\mathcal{B}_2)u_{N-1}}{\Delta x^2}+\underbrace{\frac{7\mathcal{B}_2}{3\Delta x}}_{v^\prime(x_{N-1})}.
\end{gathered}
\end{eqnarray} 

This leads to the semi-discrete system $d\vec{u}/dt = A\vec{u} + \vec{v}$,
with nonzero $\vec{v}$ near the boundary. The matrix $A$ has the form:

\begin{equation}
\label{Eq.matirxA}
A = \frac{1}{\Delta x^2}
\left(
\begin{array}{ccccccc}
7\mathcal{A}_1 \Delta x / 3 - 5/2 & 8/3 & -1/6 & 0 & 0 & \cdots & 0 \\
4/3 - \mathcal{A}_1 \Delta x / 6 & -31/12 & 4/3 & -1/12 & 0 & \cdots & 0 \\
-1/12 & 4/3 & -5/2 & 4/3 & -1/12 & \cdots & 0 \\
0 & -1/12 & 4/3 & -5/2 & 4/3 & \cdots & 0 \\
\vdots & \vdots & \vdots & \vdots & \vdots & \ddots & \vdots \\
0 & \cdots  & 4/3 & -5/2 & 4/3 & -1/12 &0 \\
0 & \cdots  & -1/12 & 4/3 & -5/2 & 4/3 & -1/12 \\
0 & \cdots  & 0 & -1/12 & 4/3 & -31/12 & 4/3+\mathcal{B}_1\Delta x/6 \\
0 & \cdots  & 0 & 0 & -1/6 & 8/3 & -5/2-7\mathcal{B}_1\Delta x/3 \\
\end{array}
\right)
\end{equation}

The matrix is nearly pentadiagonal but changes structure in the final
rows due to the imposed boundary conditions.
\end{example}

Next, the system is recast into a homogeneous form following the approach in Ref.~\cite{jin2024schr}::
\begin{eqnarray}
\label{Eq.homogeneous transformation}
\begin{gathered}
\frac{d}{dt}
\underbrace{
\left[
\begin{array}{c}
\vec{u} \\
\vec{r}
\end{array}
\right]
}_{\vec{w}}
= 
\underbrace{
\left[
\begin{array}{cc}
A & B \\
0 & 0
\end{array}
\right]
}_{S}
\left[
\begin{array}{c}
\vec{u} \\
\vec{r}
\end{array}
\right]
;\qquad
\left[
\begin{array}{c}
\vec{u}(0) \\
\vec{r}(0)
\end{array}
\right]
 = 
\left[
\begin{array}{c}
\vec{u}_0 \\
\vec{r}_0
\end{array}
\right]
;\quad
\vec{r} = \vec{r}(0) = \frac{1}{\sqrt{2^{n}}}(1, 1, \dots, 1)^\top;\\ B= \sqrt{2^{n}}\times \mathrm{diag}\{v(x_0)+v^\prime (x_0),v(x_1)+v^\prime (x_1),\dots,v(x_{N-1})+v^\prime(x_{N-1})\};\\
\frac{d\vec{w}}{dt} = S\vec{w}, \qquad \vec{w}(0) = \vec{w}_0.
\end{gathered}\label{eq: homogeneous PDE after Schrodingerisation}
\end{eqnarray}

To transform this system into a Hamiltonian simulation problem, a corresponding Hamiltonian $H$ is derived from the matrix $S$. This is achieved via the method of
Schr\"odingerisation~\cite{PRL2024, PRA2023,analog,cao2023quantum,PRS2024},
which converts the non-Hermitian generator
$S = S_1 + i S_2$ into a Hermitian operator $H$
acting on an extended Hilbert space.

We introduce an auxiliary mode indexed by $\xi$ and
allocate $n_\xi \approx n$ qubits to represent its
coordinate operator. The transformation takes the form:
\begin{eqnarray}
\begin{gathered}
    S(t) \rightarrow H(t) = S_1(t) \otimes \mathbf{x}_{\xi} 
    + S_2(t) \otimes \mathbf{1}_{\xi}\\
    \mathbf{x}_{\xi}=diag\{a_\xi,a_\xi+\frac{b_\xi-a_\xi}{2^{n_\xi}-1},a_\xi+\frac{b_\xi-a_\xi}{2^{n_\xi}-1}\times2,\dots,b_\xi\};\quad \mathbf{1}_{\xi}=diag\{1,\dots,1\},
\end{gathered}\label{eq: Schrodingerisation with time}
\end{eqnarray}
where $\mathbf{x}_{\xi}$ is a $n$-qubit matrix representation of the quantum coordinate operator on the interval $[a_\xi,b_\xi]$; commonly $b_\xi=-a_\xi\approx10$, see \cite{PRL2024}. Here $S_1(t) = (S(t) + S^{\dagger}(t))/2 = S_1^{\dagger}$ corresponds to 
the completely Hermitian part of $S(t)$ and $S_2(t) = (S(t) - S^{\dagger}(t))/2i
= S_2^{\dagger}(t)$ is associated with the completely anti-Hermitian 
part of $S(t)$. Thus, the homogenized system is mapped to the Schr\"odinger
form:
\begin{eqnarray}
    \frac{d\psi}{dt}=iH(t)\psi;\quad\psi(0)=\vec{\omega}_0\otimes\frac{2}{1+\vec{x}_\xi^2};\quad \vec{x}_\xi=\{a,a+\frac{b-a}{2^{n_\xi}-1},a+\frac{b-a}{2^{n_\xi}-1}\times2,\dots,b\}^T,
    \label{eq:Scrodinger_equation_for_homogenious_system}
\end{eqnarray}
where the symbol $T$ in $\{\dots\}^T$ denotes transposition of an array, and the expression $\frac{2}{1 + \vec{x}_\xi^2}$ represents an element-wise transformation applied to the vector $\vec{x}_\xi$. The choice of $n_\xi$ is discussed in~\cite{schrodingerisation_optimal_queries}, where the authors suggest a method to slightly modify the initial condition in the auxiliary $\xi$-dimension to achieve an error scaling of $\mathcal{O}\big(\log(1/\epsilon)\big)$.

To eliminate the time dependence in \( H(t) \), we apply  
a transformation from Ref.~\cite{cao2023quantum}, which converts the  
time-dependent Hamiltonian \( H(t) \) into an enlarged time-independent  
Hamiltonian \( H^\prime \) by introducing an additional mode  
labeled \( s \). This requires approximately \( n_s \approx n \)  
additional qubits:
\begin{eqnarray}
    H(t) \rightarrow H^\prime = \textbf{p}_s \otimes \mathbf{1} + H(t\rightarrow\textbf{x}_s);\quad \textbf{p}_s=-\frac{i}{2\Delta s}\left(\begin{array}{ccccccc}
    0 & 1 & 0& \dotsm &0& 0 & -1 \\
    -1 & 0 & 1&\dotsm& 0&0 & 0\\
    \rotatebox[origin=c]{270}{\dots}&&&\rotatebox[origin=c]{-45}{\dots}&&&\rotatebox[origin=c]{270}{\dots}\\ 
    0 & 0&0 &\dotsm &-1& 0 &1\\
    1 & 0&0 &\dotsm &0& -1 &0\\
    \end{array}
    \right),
    \label{eq: get rid of time dependence}
\end{eqnarray}
where $\textbf{p}_s$ is a discretized $2^{n_s}\times 2^{n_s}$ matrix version of the  
momentum operator $\hat{p} = -i \partial / \partial x$,  
constructed using the central difference scheme; and  
$\textbf{x}_s \sim \textbf{x}_\xi$ denotes a discretized  
$2^{n_s}\times 2^{n_s}$ matrix version of the coordinate operator. The grid parameter  
$\Delta s$ is analogous to the classical time step  
$\Delta t$ used in numerical PDE methods~\cite{zienkiewicz2005finite,ozicsik2017finite}.  

Thus, the original problem~(\ref{eq:PDE discretized}) is reduced to a Hamiltonian simulation with a time-independent Hamiltonian $H^\prime$:  

\begin{eqnarray}
    \frac{d\phi}{dt}=iH^\prime \phi;\quad \phi(0)=\vec{w_0}\otimes \frac{2}{1+\vec{x}_\xi^2}\otimes \delta_s(\vec{x}_s),
    \label{eq: discretized hamiltonian simulation of time independent problem}
\end{eqnarray}
where the initial condition for the introduced mode $s$ is the discrete matrix version of the delta function. Commonly, as an approximation of the delta function one uses Gaussian distribution with a small standard deviation $\sigma\rightarrow0$:
\begin{equation}
    \delta_s(\vec{x}_s)\approx \left( \frac{1}{2 \pi \sigma^2} \right)^{1/4} \exp\left({-\frac{\vec{x}_s^2}{4 \sigma^2}}\right).
\end{equation}
The use of a squeezed Gaussian to approximate the delta function  
introduces a controllable error in the resulting quantum state which is described in \cite{cao2023quantum, cao2024unifying}.  

Let $\rho_\sigma(t)$ denote the reduced mixed state obtained after  
tracing out the clock register from the evolved state $\phi(t)$, see Fig.~\ref{fig:Schrodingerisation}. The resulting trace-norm error between the  
ideal pure state and its mixed-state approximation scales as  
\begin{eqnarray}  
\left\| \ket{\psi(t)}\bra{\psi(t)} - \rho_\sigma(t) \right\|_{\text{tr}}  
\leq \sqrt{C\sigma},  
\end{eqnarray}  
where $C$ is a constant depending on system properties.  
To further reduce this error, one can use Richardson extrapolation  
with widths $\sigma / k_j$ and weights $\alpha_j$, leading to  
\begin{eqnarray}  
\left| \sum_{j=1}^M \alpha_j \Tr\left( \hat{O} \rho_{\sigma / k_j}(t) \right)  
- \bra{\psi(t)} \hat{O} \ket{\psi(t)} \right|  
\leq 2 \sqrt{C_m} \sigma^m + \mathcal{O}(\sigma^{m+2}),  
\end{eqnarray}  
where $C_m$ depends on the observable $\hat{O}$ and the extrapolation  
order $m$. These bounds confirm that the delta-function approximation  
introduces only polynomial error and does not compromise simulation  
efficiency.

\begin{figure}[H]
    \centering
    \includegraphics[width=0.7\linewidth]{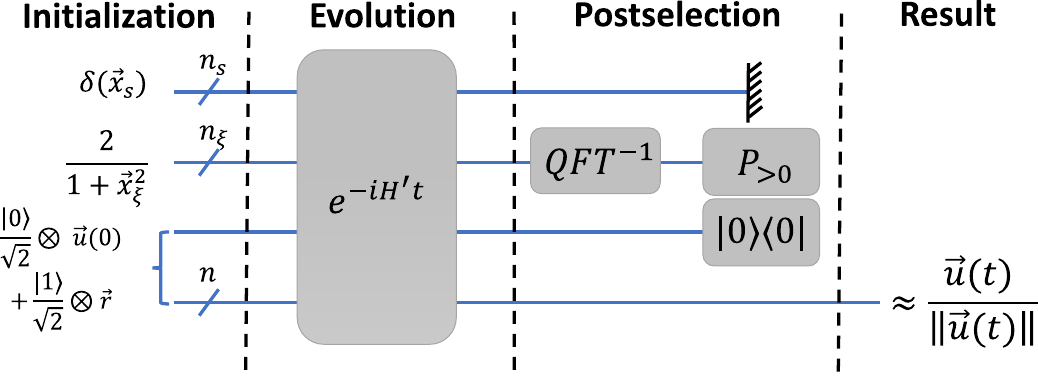}
    \caption{The overall procedure for numerically simulating the PDE of  
the form~(\ref{eq: PDE 1D MAIN})  
requires four quantum registers (from bottom to top):  
(I) an $n$-qubit register to store the main variable  
$u(x,t)$ of the PDE~(\ref{eq:PDE discretized});  
(II) a $1$-qubit register used to convert the inhomogeneous  
PDE~(\ref{eq:PDE discretized}) into the homogeneous form  (\ref{eq: homogeneous PDE after Schrodingerisation});  
(III) an $n_\xi$-qubit register to map the non-conservative  
system~(\ref{eq: homogeneous PDE after Schrodingerisation}) into the conservative  
Schrödinger equation~(\ref{eq:Scrodinger_equation_for_homogenious_system});  
(IV) an $n_s$-qubit register to transform the time-dependent  
Schrödinger equation~(\ref{eq:Scrodinger_equation_for_homogenious_system})  
into a time-independent system as in  
(\ref{eq: discretized hamiltonian simulation of time independent problem}).  
The full algorithm consists of three parts:  
\textbf{Initialization}, where all registers are prepared in the  
initial state $\phi(0)$ defined in~(\ref{eq: discretized hamiltonian simulation of time independent problem});  
\textbf{Evolution}, where the unitary $e^{iH^\prime t}$ is applied  
to obtain the evolved state $\phi(t)$;  
and \textbf{Postselection}, as described in~\cite{cao2023quantum}.  
During postselection, the second register is measured in state  
$\ket{0}$. On the third register, a reverse quantum Fourier  
transform~\cite{nielsen2002quantum} is applied, followed by  
a measurement of any positive value $P_{>0}$ of coordinate  
$\xi$ (i.e., a computational basis state  
$\ket{\kappa}^{n_\xi}$ such that $\vec{x}_\xi[\kappa] > 0$).  
Finally, the fourth register is traced out.  
As a \textbf{Result}, the final state is a mixed quantum state  
$\rho$ that closely approximates the normalized pure target state  
$\vec{u}(t)/\|\vec{u}(t)\|$.
}\label{fig:Schrodingerisation}
\end{figure}

From the problem~(\ref{eq: discretized hamiltonian simulation of time independent problem}),  
if time-independent quantum simulation of $H^\prime$ is performed,  
i.e., $\phi(t) = e^{iH^\prime t} \phi(0)$, one can efficiently  
recover the numerical solution of the original PDE  
(\ref{eq: PDE 1D MAIN})  
as a quantum state. The overall procedure is carefully detailed  
in~\cite{cao2023quantum,PRL2024} and is schematically illustrated  
in Fig.~\ref{fig:Schrodingerisation}.  

Despite the introduction of auxiliary registers, the quantum algorithm circumvents the classical curse of dimensionality. Classically, increasing dimension  
leads to exponential scaling: $N^d$ points for $d$ variables.  
On a quantum computer, each additional register requires only $\log_2 N$ qubits. Thus, all the added dimensions—used for homogenization, Schrödingerisation,  
and time embedding—grow the circuit depth and space linearly.  
The resulting Hamiltonians remain sparse and can be efficiently implemented.  

As will be demonstrated, the total complexity scales linearly with the number of physical and auxiliary dimensions and polylogarithmically with spatial resolution.
This compact encoding scheme enables efficient simulation of high-dimensional PDEs that are classically intractable.

\section{Block-encoding of the Hamiltonian}\label{sec: 1D}

In this section, we provide explicit instructions on how to efficiently construct the block-encoding of $H^\prime$ introduced in the previous section; see Eq.~\eqref{eq: get rid of time dependence}. To simplify the presentation, we now assume that the functions $f_k(x,t)$ in the original PDE~\eqref{eq: PDE 1D MAIN} are independent of time $t$; thus, we focus on the time-independent Hamiltonian $H$ as presented in Eq.~(\ref{eq: Hamiltonian main}). Before proceeding to the implementation details in this section, readers may find it helpful to refer to Appendix~\ref{appendix: important definitions}, where we introduce important mathematical definitions and quantum computational concepts. Later, in Sec.~\ref{sec: multidim}, we clarify how to extend the method to handle time-dependent functions and the multidimensional case.
\begin{eqnarray}
\begin{gathered}
    \frac{d\psi}{dt}=iH\psi;\quad H=S_1 \otimes \mathbf{x}_{\xi} 
    + S_2 \otimes \mathbf{1}_{\xi};\\
    S_1=\left(\begin{array}{cc}
    \frac{A+A^\dagger}{2} & B/2 \\
    B/2 & 0 \\    
    \end{array}
    \right);\qquad S_2=\left(\begin{array}{cc}
    \frac{A-A^\dagger}{2i} & B/2i \\
    -B/2i & 0 \\    
    \end{array}
    \right).
\end{gathered}
\label{eq: Hamiltonian main}
\end{eqnarray}
One of the main goals of this paper is to build block-encodings for $S_1$ and $S_2$ using the simplest quantum computational units: one-qubit rotations and CNOT gates. Using those block-encodings as an input for the quantum singular value transformation we achieve: $H\rightarrow e^{iHt}$, see Thm.~\ref{theorem: optimal block-hamiltonian simultaions}.

\subsection{Diagonal sparsity}\label{subsec: sparsity}

As clearly illustrated in Example~\ref{example: discretization},  
the matrices under consideration are sparse. In this subsection, we describe our approach to handling this sparsity.

\begin{lemma}[Banded-sparse matrix index]
\label{Lemma: Diagonal sparsity}
For any given banded sparse matrix $D$, we denote $D^{(s)}$ as the  
$s$-th non-zero element in the matrix's first row. For example:
\[
D = \left(
\begin{array}{cccccc}
2 & 3 & 0 & 0 & 0 & 4 \\
4 & 2 & 3 & 0 & 0 & 0 \\
0 & 4 & 2 & 3 & 0 & 0 \\
0 & 0 & 4 & 2 & 3 & 0 \\
0 & 0 & 0 & 4 & 2 & 3 \\
3 & 0 & 0 & 0 & 4 & 2 \\
\end{array}
\right)
=
\left(
\begin{array}{cccccc}
D^{(0)} & D^{(1)} & 0 & 0 & 0 & D^{(2)} \\
D^{(2)} & D^{(0)} & D^{(1)} & 0 & 0 & 0 \\
0 & D^{(2)} & D^{(0)} & D^{(1)} & 0 & 0 \\
0 & 0 & D^{(2)} & D^{(0)} & D^{(1)} & 0 \\
0 & 0 & 0 & D^{(2)} & D^{(0)} & D^{(1)} \\
D^{(1)} & 0 & 0 & 0 & D^{(2)} & D^{(0)} \\
\end{array}
\right).
\]
\end{lemma}

\begin{remark}
We extend the definition of the banded-sparse matrix index to cases  
where a sparse $N\times N$ matrix $D$ becomes a banded matrix upon  
substituting all non-zero elements with $1$, reflecting the band  
structure of non-zero elements. Thus, the notation $D^{(s)}_i$ means  
that we first address the $s$-th non-zero element in the first row,  
then add $i$ (modulo $N$) to get the column index. Here is an example:
\[
D = \left(
\begin{array}{cccccc}
2 & 6 & 0 & 0 & 0 & 4 \\
1 & 2 & 3 & 0 & 0 & 0 \\
0 & 6 & 7 & 7 & 0 & 0 \\
0 & 0 & 1 & 6 & 2 & 0 \\
0 & 0 & 0 & 3 & 8 & 9 \\
9 & 0 & 0 & 0 & 5 & 2 \\
\end{array}
\right)
=
\left(
\begin{array}{cccccc}
D^{(0)}_0 & D^{(1)}_0 & 0 & 0 & 0 & D^{(2)}_0 \\
D^{(2)}_1 & D^{(0)}_1 & D^{(1)}_1 & 0 & 0 & 0 \\
0 & D^{(2)}_2 & D^{(0)}_2 & D^{(1)}_2 & 0 & 0 \\
0 & 0 & D^{(2)}_3 & D^{(0)}_3 & D^{(1)}_3 & 0 \\
0 & 0 & 0 & D^{(2)}_4 & D^{(0)}_4 & D^{(1)}_4 \\
D^{(1)}_5 & 0 & 0 & 0 & D^{(2)}_5 & D^{(0)}_5 \\
\end{array}
\right)
\]
\end{remark}

Based on the introduced sparse access, we define a quantum operation  
that transforms the banded-sparse index into the column index of a  
non-zero matrix element. To understand the methodology for counting gate complexity and ancilla requirements across multiple results in this work, we refer the reader to Appendix~\ref{appendix: multi-control}, where we outline the general resource accounting framework.

\begin{lemma}[Banded-sparse-access-oracle and resource cost (Lemma 1 from \cite{guseynov2024efficientPDE})]
\label{lemma: Banded-sparse-access}
Let $D$ be a $2^l$-sparse $2^n \times 2^n$ matrix whose non-zero elements,  
when replaced by $1$, form a banded matrix. Define the banded-sparse-access-oracle $\hat{O}^{BS}_D$ as the unitary acting as
\[
\hat{O}^{BS}_D\ket{0}^{n-l}\ket{s}^l\ket{i}^n := \ket{r_{si}}^n \ket{i}^n,
\]
where $r_{si} = r_{s0} + i \mod 2^n$ represents the $(s)$-th  
banded-sparse matrix index ($s$-th non-zero element in the $i$-th line). This oracle $\hat{O}^{BS}_D$ can be implemented with at most:
\begin{enumerate}
    \item $(2^l + 1)(32n - 48)$ single-qubit operations;
    \item $25 \cdot 2^l n - 36 \cdot 2^l + 32n - 48$ CNOT gates;
    \item $n - 1$ pure ancilla qubits.
\end{enumerate}
\end{lemma}

\begin{remark}\label{remark: sparsity maximum}
    Note that the Banded-sparse-access-oracle with the largest diagonal sparsity index $s_{\max}=7$ can be constructed for matrix $A$ from Example~\ref{example: discretization}; thus, the matrix contains $7$ diagonals with non-zero elements (including $2$ additional diagonals arising from boundary effects). Later on we use Banded-sparse-access-oracles to switch between sparse indexes notation to normal notation and vice-versa. 
\end{remark}

\subsection{Sparse-amplitude-oracle for the derivative operator with periodic boundaries}

In this subsection, we achieve another important milestone toward constructing the block-encoding of $S_1$ and $S_2$ from (\ref{eq: Hamiltonian main}) which is Sparse-amplitude-oracle for Banded-sparse matrix as in Lemma~\ref{Lemma: Diagonal sparsity}. Note that, in such matrices, the non-zero elements remain constant across each row. For example, the discretized derivative operator (which is important in this paper) for the periodic boundaries has a form similar to:
\begin{eqnarray}
\frac{\partial}{\partial x}\approx
    \frac{1}{2\Delta x}\left(\begin{array}{ccccccc}
    0 & 1 & 0& \dotsm &0& 0 & -1 \\
    -1 & 0 & 1&\dotsm& 0&0 & 0\\
    \rotatebox[origin=c]{270}{\dots}&&&\rotatebox[origin=c]{-45}{\dots}&&&\rotatebox[origin=c]{270}{\dots}\\ 
    0 & 0&0 &\dotsm &-1& 0 &1\\
    1 & 0&0 &\dotsm &0& -1 &0\\
    \end{array}
    \right).
    \label{eq: derivative discretize periodic boundaries}
\end{eqnarray}
The exact structure depends on the finite-difference scheme employed (see Table~\ref{table: famous schemes}); however, we emphasize that, for periodic boundaries, the discretized derivative operator $\frac{\partial^m}{\partial x^m}$ always corresponds to the banded-sparse matrix scenario outlined in Lemma~\ref{lemma: Banded-sparse-access}.

\begin{lemma}[Sparse-amplitude-oracle for a banded-sparse matrix (Lemma 3 from \cite{guseynov2024efficientPDE})]
\label{lemma: amplitude-oracle for a banded-sparse matrix}
Let \( D \) be a $2^l$-sparse banded matrix of size $2^n \times 2^n$,  
where the non-zero values in each row follow a shift-invariant pattern  
(as in Lemma~\ref{Lemma: Diagonal sparsity}).  
Define the Sparse-amplitude-oracle for $D$ as
\begin{equation}
\hat{O}^S_{D} \ket{0}^1 \ket{s}^l := \frac{D^{(s)}}{\mathcal{N}_D} \ket{0}^1 \ket{s}^l + \sqrt{1 - \frac{|D^{(s)}|^2}{\mathcal{N}^2_D}} \ket{1}^1 \ket{s}^l,
\label{eq: amplitude_oracle_D}
\end{equation}
where $s \in \{0, \dots, 2^l-1\}$ indexes the non-zero elements in  
each row and $\mathcal{N}_D \geq \|D\|_{\max}$ is a known upper  
bound on the squared norm of the matrix entries. This oracle prepares  
a superposition where the amplitude of the $\ket{0}$ outcome encodes  
the normalized matrix value $D^{(s)}$.

Then, the unitary $\hat{O}^S_D$ can be implemented using:
\begin{enumerate}
    \item $2^l$ single-qubit operations;
    \item $2^l$ CNOT gates.
\end{enumerate}
\end{lemma}

\begin{remark}
This Sparse-amplitude-oracle is particularly used for  
encoding higher-order derivative operators with periodic boundary  
conditions, such as $\partial^m/\partial x^m$ arising from finite-difference  
formulas(see Table.~\ref{table: famous schemes}); since these matrices maintain an identical banded pattern across all rows, the oracle construction remains valid for arbitrary derivative order $m$. 
\end{remark}

\subsection{Block-encoding of $B$ and Amplitude-oracle for $f(x)$}\label{sub: block-encoding of B}

Our next objective is to construct the amplitude oracle for the  
functions \( f(x) \) and \( v(x) \), as they appear in  
the original PDE (Eq.~\ref{eq: PDE 1D MAIN}). In this paper,  
we assume that these functions are continuous or at least  
piece-wise continuous, allowing the application of the following  
Theorem. This assumption is physically motivated, as many real-world systems evolve through smooth, continuous processes, making this setting both practical and  theoretically well-justified.  

Specifically, the diagonal entries of the matrix \( B \)  
correspond to a discretized representation of the function  
\( v(x) \), given by \( B_{ii} = v(x_i) + v'(x_i) \), where the corrections \( v^\prime(x_i) \) account for the enforcement of boundary conditions, see Example~\ref{example: discretization}. In parallel, we employ a block-encoded representation of the function
\( f(\mathbf{x}) \), implemented as a diagonal matrix with  
entries \( f(x_i) \) along the diagonal.

\begin{theorem}[Amplitude-oracle for piece-wise polynomial function (Thm.~5 from \cite{guseynov2024explicit_Quantum_state_prep})]\label{Theorem: Amplitude-oracle for piece-wise polynomial function}
Let $f(x)$ be a piece-wise continuous function on the interval $x\in[a,b]$. where $\abs{a},\abs{b}<1$ and $f: \mathbb{R} \rightarrow \mathbb{C}$ with polynomial decomposition
\[
f(x)=\begin{cases} 
f_1(x) = \sum_{i=0}^{Q_1} \sigma_i^{(1)} x^i = \alpha_1^{(1)} P_1^{(1)}(x) + \alpha_2^{(1)} P_2^{(1)}(x) + i \alpha_3^{(1)} P_3^{(1)}(x) + i \alpha_4^{(1)} P_4^{(1)}(x), & \text{if } a \leq x \leq K_1 \\
f_2(x) = \sum_{i=0}^{Q_2} \sigma_i^{(2)} x^i = \alpha_1^{(2)} P_1^{(2)}(x) + \alpha_2^{(2)} P_2^{(2)}(x) + i \alpha_3^{(2)} P_3^{(2)}(x) + i \alpha_4^{(2)} P_4^{(2)}(x), & \text{if } K_2 \geq x > K_1 \\
\multicolumn{2}{c}{\vdots} \\
f_G(x) = \sum_{i=0}^{Q_G} \sigma_i^{(G)} x^i = \alpha_1^{(G)} P_1^{(G)}(x) + \alpha_2^{(G)} P_2^{(G)}(x) + i \alpha_3^{(G)} P_3^{(G)}(x) + i \alpha_4^{(G)} P_4^{(G)}(x), & \text{if } b \geq x > K_{G-1}
\end{cases};
\]
\[\sigma_i^{(j)}\in\mathbb{C};\quad \alpha_s^{(j)}\in\mathbb{R},\]
where polynomials $P^{(j)}_s(x)$ satisfy:
\begin{itemize}
    \item $P^{(j)}_s(x)$ has parity-($s$ mod 2);
    \item $P^{(j)}_s(x)\in\mathbb{R}[x]$;
    \item For all $x\in[-1,1]$: $\abs{P^{(j)}_s}<1$.
\end{itemize}

Then, we can construct a unitary operation $\hat{O}_f$ which is a block-encoding of a discretized ($\Delta x=1/2^n$) version of $f(\textbf{x})$ with $\textbf{x}=diag\{a,a+\frac{b-a}{2^n-1},a+\frac{b-a}{2^n-1}\times2,\dots,b\}$
\begin{eqnarray}
    \begin{gathered}
       \hat{O}_f\ket{0}^{\lceil\log_2n\rceil+\lceil\log_2G\rceil+3}\ket{j}^n=\frac{f(a+\frac{b-a}{2^n-1}\times j)}{\mathcal{N}_f}\ket{0}^{\lceil\log_2n\rceil+\lceil\log_2G\rceil+3}\ket{j}^n+J^{(j)}_f\ket{\bot_{\ket{0}}}^{n+\lceil\log_2n\rceil+\lceil\log_2G\rceil+3}\\
       =\frac{f(x_j)}{\mathcal{N}_f}\ket{0}^{\lceil\log_2n\rceil+\lceil\log_2G\rceil+3}\ket{j}^n+J^{(j)}_f\ket{\bot_{\ket{0}}}^{n+\lceil\log_2n\rceil+\lceil\log_2G\rceil+3},
    \end{gathered}\label{eq:coordinate oracle}
\end{eqnarray}
where $N_f\sim \max_{x\in{[a,b]}}f(x)$; $\ket{\bot_{\ket{0}}}^{n+\lceil\log_2n\rceil+\lceil\log_2G\rceil+3}$ denotes a state orthogonal 
to $\ket{0}^{\lceil\log_2n\rceil+\lceil\log_2G\rceil+3} \otimes I^{\otimes n}$ meaning
\[
\forall\ket{\psi}^n:\bra{\bot_{\ket{0}}}^{n+\lceil\log_2n\rceil+\lceil\log_2G\rceil+3}
\left(\ket{0}^{\lceil\log_2n\rceil+\lceil\log_2G\rceil+3}\otimes\ket{\psi}^n\right)=0.
\]
The second term involving $J_f$ corresponds to the orthogonal component introduced by the block-encoding structure and ensures unitarity of the oracle $\hat{O}_f$.

The total resource cost for implementing $n$-qubit version of $\hat{O}_f$ as in Eq.~(\ref{eq:coordinate oracle}) scales as
\begin{enumerate}
    \item $\mathcal{O}(\sum_{g=1}^GQ_gn\log n)$ quantum gates;
    \item $n-1$ pure ancillas (see Definition~\ref{def:pure ancilla}).
\end{enumerate}
\end{theorem}

Applying Thm.~\ref{Theorem: Amplitude-oracle for piece-wise polynomial function} 
to the function \( v(x) + v'(x) \) enables us to 
construct a \( (\mathcal{N}_B, \lceil \log_2 n \rceil + \lceil\log_2G_v\rceil+3, 0) \)-block-encoding 
of the matrix \( B \), denoted by \( \hat{O}_B \). Here, 
the normalization factor satisfies \( \mathcal{N}_B \sim \sqrt{2^n} \max_{x \in [a,b]} v(x) \), 
as discussed in Eq.~(\ref{eq: homogeneous PDE after Schrodingerisation}) and 
Definition~\ref{def:block-encoding}.

Similarly, we can construct $\hat{O}_f$ which is a \( (\mathcal{N}_f, \lceil \log_2 n \rceil + \lceil\log_2G_f\rceil+3, 0) \)-block-encoding of $f(\mathbf{x})$ for each term from the Eq.~\ref{eq: PDE 1D MAIN} with $\mathcal{N}_f\sim \max_{x\in[a,b]}f(x)$.

\subsection{Block-encoding for the periodic boundary case}

In this subsection, we describe the implementation of the block-encoding  
for a single term of the discretized operator \(A\) from Eq.~(\ref{eq:PDE discretized})  
in the case of periodic boundary conditions. The construction was first developed  
in our previous work~\cite{guseynov2024efficientPDE} for operators of the form  
\(A_k = f(x)\frac{\partial^m}{\partial x^m}\), where \(m\) is the derivative order  
and \(f(x)\) is a piecewise continuous function. Here, we revisit the periodic  
case with several minor refinements, which serve as a baseline for further  
generalization to other types of boundary conditions.

First, we discretize the  
operator $A_k$ to map it to an \( n \)-qubit system. (i) The function  
\( f(x) \) is represented as a \( 2^n \times 2^n \) diagonal  
matrix, where the diagonal elements are given by  
\( [f(\textbf{x})]_{ii} = f(x_i) \). (ii) The derivative operator is  
assumed to correspond to a periodic boundary problem, so its  
discretized form is a \( 2^n \times 2^n \) matrix, similar  
to the form described in Eq.~(\ref{eq: derivative discretize periodic boundaries}). 

\begin{figure}[H]
    \centering
    \includegraphics[width=0.8\linewidth]{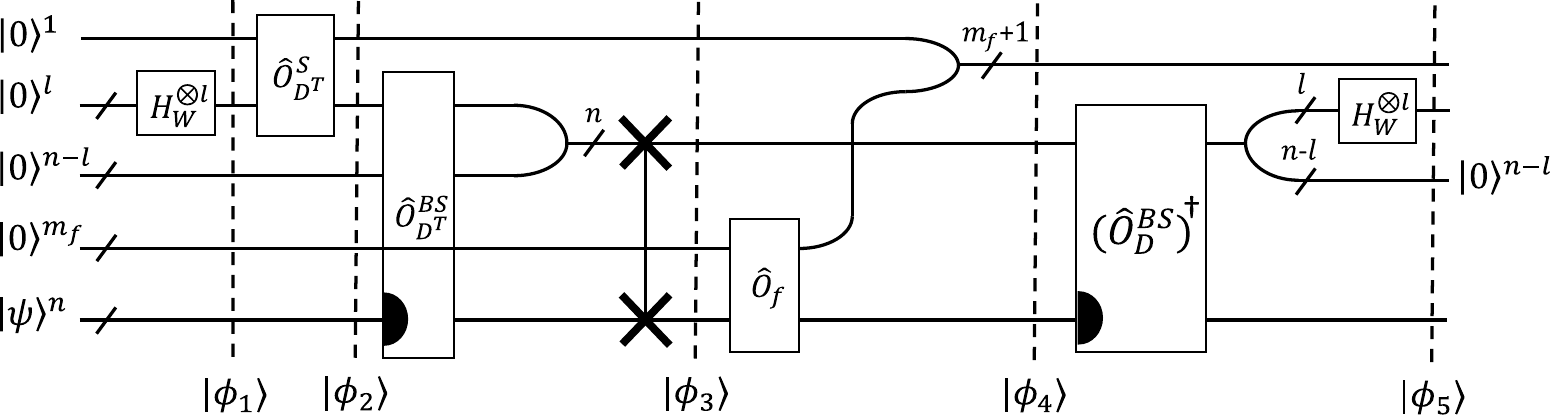}
    \caption{Quantum circuit for the block-encoding of the discretized operator \( A_k = f(x) \frac{\partial^m}{\partial x^m} \) for the case of periodic boundaries (banded $D$). The wave functions \( \ket{\phi_i} \) are detailed in Eq.~(\ref{eq: clarification of wave functions for periodic boundaries case}). The circuit includes a one-qubit upper register and an \( l \)-qubit quantum register (for sparse indexing), which are used for the Sparse-amplitude-oracle for the transposed derivative operator \( \hat{O}_{D^T}^S \) as described in Lemma~\ref{lemma: amplitude-oracle for a banded-sparse matrix}. Here, \( T \) denotes the matrix transposition operation. The operator \( \hat{O}_{D^T}^{BS} \) is the Banded-sparse-access-oracle, as defined in Lemma~\ref{lemma: Banded-sparse-access}, with the filled semicircle representing the \( n \)-qubit quantum register that remains unchanged under the action of \( \hat{O}_{D^T}^{BS} \). The \( m_f = \lceil\log_2 n \rceil + \lceil \log_2 G_f \rceil + 3 \) qubit register is reserved for the amplitude oracle \( \hat{O}_f \), which implements the piecewise polynomial function from Thm.~\ref{Theorem: Amplitude-oracle for piece-wise polynomial function}. The two connected crosses depict the SWAP operation \cite{nielsen2002quantum} between two \( n \)-qubit registers. The operator \( (\hat{O}_D^{BS})^\dagger \) is the Hermitian conjugate of \( \hat{O}_D^{BS} \), returning the upper \( n \)-qubit register to the sparse indexing state. The third \( n - l \)-qubit register is initialized in the \( \ket{0}^{n-l} \) state before and after the unitary operation \( U_{A_k} \) remain in the state $ \ket{0}^{n-l} $, meaning these qubits serve as pure ancillas, as defined in Definition~\ref{def:pure ancilla}. The Y-shaped frame represents the simple merging of registers without any operations. At certain points in the quantum circuit, the \( m_f \) qubits are moved upwards for convenience, as indicated by the zigzag line. $\hat{O}_{D^T}^{BS}$ does not act on $m_f$-qubit quantum register; that is why the corresponding wire goes above the box.}
    \label{fig: 1 term PERIODIC}
\end{figure}

In Fig.~\ref{fig: 1 term PERIODIC}, we illustrate our strategy for constructing a \( (2^l \mathcal{N}_f \mathcal{N}_D, \lceil \log_2 n \rceil + \lceil \log_2 G_f \rceil + l + 4, 0) \)-block-encoding of \( A_k \), with the scheme further clarified in Eq.~(\ref{eq: clarification of wave functions for periodic boundaries case}). For simplicity, we assume that the diagonal sparsity of the derivative operator \( D \) is of degree \( 2 \), i.e., \( s \in \{0, \dots, 2^l - 1\} \); however, this construction can be generalized to arbitrary diagonal sparsity by replacing \( H_W^{\otimes l} \rightarrow H_W^{(\kappa)} \), where

\begin{eqnarray}
    H_W^{(\kappa)}\ket{0}^{\lceil \log_2 \kappa \rceil} = \frac{1}{\sqrt{\kappa}} \sum_{s=0}^{\kappa - 1} \ket{s}^{\lceil \log_2 \kappa \rceil}.
    \label{eq: arbitrary sparcity}
\end{eqnarray}

This operation can be implemented with gate complexity \( \mathcal{O}(\log \kappa) \), as demonstrated in Ref.~\cite{shukla2024efficient}.

The total resource cost of the circuit presented in Fig.~\ref{fig: 1 term PERIODIC} scales as:
\begin{enumerate}
    \item \( \mathcal{O} \left( \sum_{g=1}^G Q_g n \log n + 2^l n \right) \) quantum gates;
    \item \( 2n - l \) pure ancillas.
\end{enumerate}

The following equation clarifies the structure of Fig.~\ref{fig: 1 term PERIODIC}. As input, we use an arbitrary quantum state \( \ket{\psi}^n = \sum_{j=0}^{2^n - 1} \sigma_j \ket{j}^n \). The index \( s \) refers to the diagonal sparsity index (see Lemma~\ref{Lemma: Diagonal sparsity}), and \( F(A_k, i) \) denotes the set of column indices corresponding to the non-zero elements of the \( i \)-th row of the discretized \( 2^n \times 2^n \) matrix \( A_k \) or \(D\).

\begin{eqnarray}
    \begin{gathered}
    \ket{\psi}^n=\sum_{j=0}^{2^n-1}\sigma_j\ket{j}^n;\\
        \ket{\phi_1}=\frac{1}{\sqrt{2^l}}\sum\limits_{\substack{s=0,\dots2^l-1\\ j=0,\dots,2^n-1}}\sigma_j\ket{0}^1\ket{s}^l\ket{0}^{n-l}\ket{0}^{m_f}\ket{j}^n;\\
        \ket{\phi_2}=\frac{1}{\sqrt{2^l}\mathcal{N}_D}\sum\limits_{\substack{s=0,\dots2^l-1\\ j=0,\dots,2^n-1}}\sigma_j(D^{T})^{(s)}\ket{0}^1\ket{s}^l\ket{0}^{n-l}\ket{0}^{m_f}\ket{j}^n+\dots;\\
        \ket{\phi_3}=\frac{1}{\sqrt{2^l}\mathcal{N}_D}\sum\limits_{\substack{i\in F(A^{(1)T},j)\\ j=0,\dots,2^n-1}}\sigma_jD_{ij}\ket{0}^1\ket{j}^n\ket{0}^{m_f}\ket{i}^n+\dots\\
        =\frac{1}{\sqrt{2^l}\mathcal{N}_D}\sum\limits_{\substack{i=0,\dots2^n-1\\ j=0,\dots,2^n-1}}\sigma_jD_{ij}\ket{0}^1\ket{j}^n\ket{0}^{m_f}\ket{i}^n+\dots;\\
        \ket{\phi_4}=\frac{1}{\sqrt{2^l}\mathcal{N}_D\mathcal{N}_f}\sum\limits_{\substack{i=0,\dots2^n-1\\ j=0,\dots,2^n-1}}\underbrace{f(x_i)D_{ij}}_{(A_k)_{ij}}\sigma_j\ket{0}^{m_f+1}\ket{j}^n\ket{i}^n+\dots\\
        =\frac{1}{\sqrt{2^l}\mathcal{N}_D\mathcal{N}_f}\sum\limits_{\substack{i=0,\dots2^n-1\\ j=F(A_k,i)}}\underbrace{f(x_i)D_{ij}}_{(A_k)_{ij}}\sigma_j\ket{0}^{m_f+1}\ket{j}^n\ket{i}^n+\dots\\
        \ket{\phi_5}=\frac{1}{2^l\mathcal{N}_D\mathcal{N}_f}\sum\limits_{\substack{i=0,\dots2^n-1\\ s=0,\dots,2^l-1}}f(x_i)D_{i}^{(s)}\sigma^{(s)}\ket{0}^{m_f+l+1}\ket{0}^{n-l}\ket{i}^n+\dots.
    \end{gathered}\label{eq: clarification of wave functions for periodic boundaries case}
\end{eqnarray}

\begin{figure}[H]
    \centering
    \includegraphics[width=0.75\linewidth]{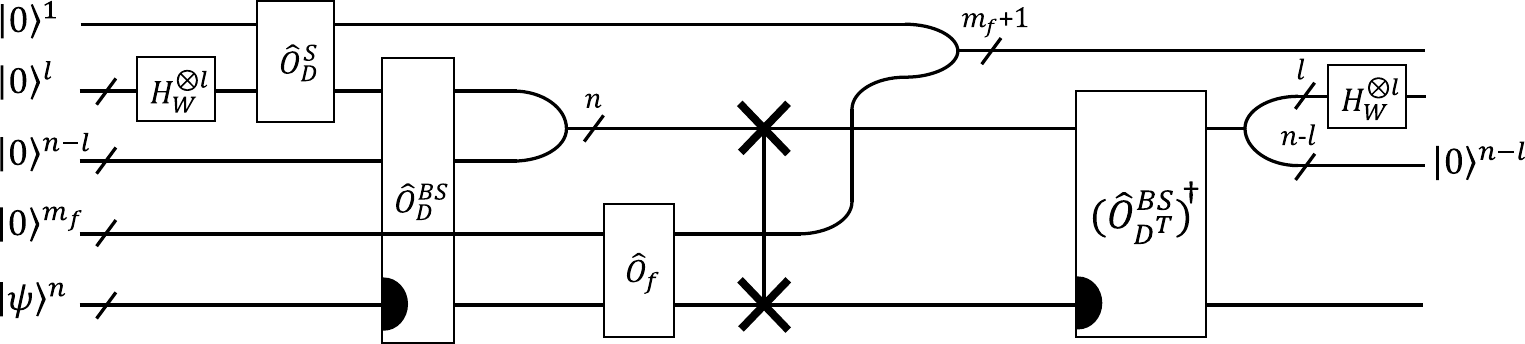}
    \caption{Quantum circuit for the block-encoding of the discretized operator with matrix elements \( (A_k)_{ij}^\dagger = f^*(x_j) D_{ji} \), where the symbol \( * \) denotes complex conjugation. For details of the original (non-conjugated) version, see Fig.~\ref{fig: 1 term PERIODIC}.}
    \label{fig:1 term PERIODIC DAGGER}
\end{figure}

Later, we require a block-encoding of the Hermitian conjugate operator \( (A_k)^\dagger \); its construction closely follows the design shown in Fig.~\ref{fig: 1 term PERIODIC} and is given in Fig.~\ref{fig:1 term PERIODIC DAGGER}.

\subsection{Generalization to Robin boundaries}

Now, we consider how to modify the quantum circuits
presented in the previous subsection to account for the
effect of Robin boundary conditions as in Eq.~(\ref{eq: PDE
1D MAIN}). The matrix $ A $, introduced in Example~\ref{example:
discretization}, differs slightly from the case with periodic
boundaries, and this difference must be incorporated into
the circuit design.

\begin{eqnarray}
\scriptsize
\begin{gathered}
    A_{\text{periodic}} \sim
    \left(
    \begin{array}{cccccccc}
 -5/2 & 4/3 & -1/12 & 0& \cdots & 0 & -1/12 &4/3 \\
 4/3 & -5/2 & 4/3 & -1/12 &0& \cdots & 0 & -1/12 \\
-1/12 & 4/3 & -5/2 & 4/3 & -1/12&0 & \cdots & 0 \\
\vdots & \vdots & \vdots & \vdots & \vdots & \ddots&\ddots & \vdots \\
0 & \cdots  &0& -1/12 & 4/3 & -5/2 & 4/3 & -1/12 \\
-1/12&0 & \cdots  &0& -1/12 & 4/3 & -5/2 & 4/3  \\
4/3 & -1/12&0& \cdots  &0& -1/12 & 4/3 & -5/2 \\ 
    \end{array}
    \right);\\[3ex]
    A_{\text{Robin}}\sim
    \left(
    \begin{array}{ccccccc}
    7\mathcal{A}_1 \Delta x / 3 - 5/2 & 8/3 & -1/6 & 0 & 0 & \cdots & 0 \\
4/3 - \mathcal{A}_1 \Delta x / 6 & -31/12 & 4/3 & -1/12 & 0 & \cdots & 0 \\
-1/12 & 4/3 & -5/2 & 4/3 & -1/12 & \cdots & 0 \\
\vdots & \vdots & \vdots & \vdots & \vdots & \ddots & \vdots \\
0 & \cdots  & -1/12 & 4/3 & -5/2 & 4/3 & -1/12 \\
0 & \cdots  & 0 & -1/12 & 4/3 & -31/12 & 4/3+\mathcal{B}_1\Delta x/6 \\
0 & \cdots  & 0 & 0 & -1/6 & 8/3 & -5/2-7\mathcal{B}_1\Delta x/3 \\
    \end{array}
    \right).
    \end{gathered}
    \label{eq: example periodic versus robin}
\end{eqnarray}

The example shown in Eq.~\ref{eq: example periodic versus robin}
illustrates a typical structure of the matrix $ A $ under periodic
and Robin boundary conditions. We observe that the primary difference between the two cases is localized near the matrix boundaries, specifically at the  
upper-left and lower-right corners. In
particular, for a central region (bulk) of indices, the two matrices coincide.
More precisely, there exist thresholds $ K_1 $ and $ K_2 $ such that
for all $ i $ satisfying $ K_1 \leq i \leq K_2 $, the corresponding
rows and columns of the matrices are identical. Formally,

\begin{eqnarray}
    \begin{gathered}
        \forall j_1,\; \exists K_1, K_2 \text{ such that }  
        K_1 \leq i \leq K_2,\; [A_{\text{periodic}}]_{i j_1} =  
        [A_{\text{robin}}]_{i j_1}; \\
        \forall j_2,\; \exists K_1, K_2 \text{ such that }  
        K_1 \leq i \leq K_2,\; [A_{\text{periodic}}]_{j_2 i} =  
        [A_{\text{robin}}]_{j_2 i},
    \end{gathered}
\end{eqnarray}

where we consider both the row-wise and column-wise
patterns. This observation implies that the internal (bulk) structure
of the operator is unaffected by the choice of boundary condition and
only a small boundary region ($\mathcal{O}(1)$ rows or columns) must
be treated differently.

Now, we introduce a special unitary operation that would allow
us to set the ancilla qubit (indicator qubit) $ c $ into the
state $ \ket{1} $ if the number of the state in the computational
basis which corresponds to the row (column) number is in the bulk region $ K_1 \leq i \leq K_2 $:
\begin{eqnarray}
    U_{\text{indic}}(K_1,K_2)\sum_{i=0}^{2^n-1}\sigma_i\ket{i}^n\ket{0}=
    \sum_{i=0}^{K_1-1}\sigma_i\ket{i}^n\ket{0}+
    \sum_{i=K_1}^{K_2}\sigma_i\ket{i}^n\ket{1}+
    \sum_{i=K_2+1}^{2^n-1}\sigma_i\ket{i}^n\ket{0}.
\end{eqnarray}
We provide an instruction with explicit circuits in the Appendix~\ref{Appendix: indicator}.
The total resources for this operation scale as:
\begin{enumerate}
    \item $\mathcal{O}(n)$ quantum gates,
    \item $n$ pure ancillas.
\end{enumerate}

Note, we always can include zeros in the set of non-zero
elements, so we now use the diagonal sparsity of matrix $ A $
$ s_{\max} = \kappa $, see Remark~\ref{remark: sparsity maximum}.

Lastly, let us consider the boundary part of the discretized derivative
operator $ D \sim \frac{\partial^m}{\partial x^m} $ from the Eq.~(\ref{eq:PDE discretized})
for Robin boundary. For this part, we address each element individually
and use controlled $ R_y $ rotations with angles:
\begin{eqnarray}
    \theta_j^s = \arccos\left(\frac{D_j^{(s)}}{\mathcal{N}_D}\right);\quad s=0,\dots,\kappa-1;\quad j=0,\dots,K_1-1,K_2+1,\dots,2^n-1,
    \label{eq:angles for Ry}
\end{eqnarray}
where $ \mathcal{N}_D \geq \|D\|_{\max} $. Thus, the generalized
quantum circuit is depicted in Fig.~\ref{fig:1 term ROBIN} and clarified
in Eq.~(\ref{eq: ROBIN clarified}); the total resource scales as:
\begin{enumerate}
    \item \( \mathcal{O} \left( \sum_{g=1}^G Q_g n \log n + \kappa\underbrace{(K_1+2^n-K_2)}_{\text{the number of deviating lines}} n \right) \) quantum gates;
    \item \( 2n \) pure ancillas.
\end{enumerate}
The total number of deviating indices, $ K_1 + 2^n - K_2 $,
can be safely assumed to be $ \mathcal{O}(1) $ as it depends
linearly on the degree of accuracy of the finite-difference
scheme used, which is a small number \cite{zienkiewicz2005finite,ozicsik2017finite}.
Thus, we can prove the following Theorem.

\begin{theorem}[One-term block-encoding]\label{theorem: 1 term robin}
    Let $ f(x) $ be a piece-wise continuous function as in Thm.~\ref{Theorem: Amplitude-oracle for piece-wise polynomial function} with number of pieces $G_f$ and complexity $Q_g$ (polynomial degree);
    and let $ 2^n \times 2^n $ matrix $ D \sim \frac{\partial^m}{\partial x^m} $
    be a discretized $\kappa$-sparse version of the derivative operator of a certain order $ m $
    for the Robin boundaries as in Eq.~(\ref{eq:PDE discretized}). Then we can construct a $(\mathcal{N}_D\mathcal{N}_f\kappa,\lceil\log_2 n
    \rceil + \lceil \log_2 G_f \rceil +\lceil \log_2 \kappa \rceil+ 4,0)$-block-encoding $U_{A_k}^{(1)}$ of the $2^n\times 2^n$ discretized version of $A_k\sim f(x)\frac{\partial^m}{\partial x^m}$ (see Eq.~(\ref{eq:PDE discretized})) with resources scaling as
    
    \begin{enumerate}
    \item \( \mathcal{O} \left( \sum_{g=1}^G Q_g n \log n + \kappa n \right) \) quantum gates;
    \item \( 2n \) pure ancillas.
    \end{enumerate}
\end{theorem}

\begin{eqnarray}
    \begin{gathered}
        \ket{\gamma_1} = \frac{1}{\mathcal{N}_D \sqrt{\kappa}} \sum_{\substack{s=0,\dots,\kappa-1 \\ 0 \leq j < K_1 \cup K_2 < j < 2^n}} \sigma_j \ket{0}^{m_f} \ket{0}^1 \ket{s}^{\lceil \log_2 \kappa \rceil} \ket{0}^{n-\lceil \log_2 \kappa \rceil} \ket{j}^n \ket{0}^1 \\
        + \frac{1}{\sqrt{\kappa}} \sum_{\substack{s=0,\dots,\kappa-1 \\ K_1 \leq j \leq K_2}} \sigma_j \ket{0}^{m_f} \ket{0}^1 \ket{s}^{\lceil \log_2 \kappa \rceil} \ket{0}^{n-\lceil \log_2 \kappa \rceil} \ket{j}^n \ket{1}^1; \\
        \ket{\gamma_2} = \frac{1}{\mathcal{N}_D \sqrt{\kappa}} \sum_{\substack{s=0,\dots,\kappa-1 \\ 0 \leq j < K_1 \cup K_2 < j < 2^n}} (D^T)^{(s)}_j \sigma_j \ket{0}^{m_f} \ket{0}^1 \ket{s}^{\lceil \log_2 \kappa \rceil} \ket{0}^{n-\lceil \log_2 \kappa \rceil} \ket{j}^n \ket{0}^1 \\
        + \frac{1}{\mathcal{N}_D \sqrt{\kappa}} \sum_{\substack{s=0,\dots,\kappa-1 \\ K_1 \leq j \leq K_2}} (D^T)^{(s)}_j \sigma_j \ket{0}^{m_f} \ket{0}^1 \ket{s}^{\lceil \log_2 \kappa \rceil} \ket{0}^{n-\lceil \log_2 \kappa \rceil} \ket{j}^n \ket{1}^1 + \dots; \\
        \ket{\gamma_3} = \frac{1}{\mathcal{N}_D \mathcal{N}_f \kappa} \sum_{\substack{s=0,\dots,\kappa-1 \\ 0 \leq j < K_1 \cup K_2 < j < 2^n}} f(x_i) (D)^{(s)}_i \sigma^{(s)} \ket{0}^{m_f+1} \ket{s}^{\lceil \log_2 \kappa \rceil} \underbrace{\ket{0}^{n-\lceil \log_2 \kappa \rceil}}_{\text{pure ancillas}} \ket{j}^n \underbrace{\ket{0}^1}_{\text{pure ancilla}} + \dots; \\
    \end{gathered}
    \label{eq: ROBIN clarified}
\end{eqnarray}

\begin{figure}[H]
    \centering
    \includegraphics[width=1\linewidth]{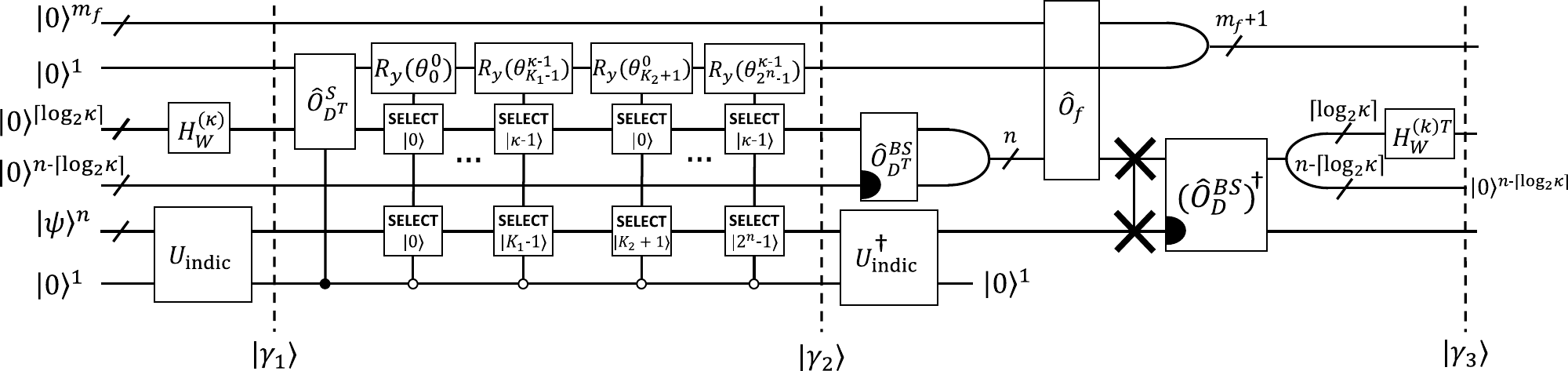}
    \caption{
    Circuit design for the block-encoding of the discretized operator
    $ A_k = f(x) \frac{\partial^m}{\partial x^m} $
    for the case of Robin boundaries (banded matrix $ D $ with some
    deviations near the boundaries). The wave functions $ \ket{\gamma_i} $
    are detailed in Eq.~(\ref{eq: ROBIN clarified}). Between the $ \gamma_1 $
    and $ \gamma_2 $ slices we address the derivative operator $ D $:
    (i) $ U_{\text{indic}} $ sets the last ancilla qubit to $ \ket{1}^1 $
    for the bulk part of $ D^T $, (ii) then for the bulk part we apply
    $ \hat{O}^S_{D^T} $, (iii) and for the boundary part we operate element-wise
    using $ R_y $ with angles $ \theta^s_j $ as in Eq.~(\ref{eq:angles for Ry}).
    The overall circuit is a generalized version of the one depicted in
    Fig.~\ref{fig: 1 term PERIODIC}. The circuit includes a one-qubit upper
    register and an $ \lceil\log_2\kappa\rceil\ $-qubit quantum register
    (for sparse indexing, see Remark~\ref{remark: sparsity maximum}), which
    are used for the Sparse-amplitude-oracle for the transposed derivative
    operator $ \hat{O}_{D^T}^S $ as described in Lemma~\ref{lemma: amplitude-oracle
    for a banded-sparse matrix}. Here, $ T $ denotes the matrix transposition
    operation. The operator $ \hat{O}_{D^T}^{BS} $ is the Banded-sparse-access-oracle,
    as defined in Lemma~\ref{lemma: Banded-sparse-access}, with the filled
    semicircle representing the $ n $-qubit quantum register that remains
    unchanged under the action of $ \hat{O}_{D^T}^{BS} $; where the matrix
    $ D $ now corresponds to Robin boundaries. The $ m_f = \lceil\log_2 n
    \rceil + \lceil \log_2 G_f \rceil + 3 $ qubit register is reserved
    for the amplitude oracle $ \hat{O}_f $, which implements the piecewise
    polynomial function from Thm.~\ref{Theorem: Amplitude-oracle for piece-wise
    polynomial function}. The two connected crosses depict the SWAP operation
    \cite{nielsen2002quantum} between two $ n $-qubit registers. The operator
    $ (\hat{O}_D^{BS})^\dagger $ is the Hermitian conjugate of $ \hat{O}_D^{BS} $,
    returning the upper $ n $-qubit register to the sparse indexing state
    and making $ n-\lceil\log_2\kappa\rceil $-qubits set back to zero state
    $ \ket{0}^1 $ (pure ancilla); thus the third $ n - \lceil\log_2\kappa\rceil $-qubit register is
    initialized in the $ \ket{0}^{n-\lceil\log_2\kappa\rceil} $ state before
    and after the unitary operation $ U_{A_k}^{(1)} $ remain in the state $ \ket{0}^{n-\lceil\log_2\kappa\rceil} $, meaning these qubits serve
    as pure ancillas, as defined in Definition~\ref{def:pure ancilla}. The Y-shaped
    frame represents the simple merging of registers without any operations.
    $ \hat{O}_{f} $ does not act on $ 1 $-qubit quantum register; that is why
    the corresponding wire goes above the box.
    }
    \label{fig:1 term ROBIN}
\end{figure}

Similarly, we can prove the corresponding result  for $U_{A_k^\dagger}$ by generalizing the circuit depicted in Fig.~\ref{fig:1 term PERIODIC DAGGER}; the resources count remains the same for both cases.

\subsection{Block-encoding for the $1$-D Hamiltonian}

We now have all the necessary ingredients to construct the block-encoding of the Hamiltonian in (\ref{eq: Hamiltonian main}). From Thm.~\ref{theorem: 1 term robin},
we implement $A_k$, representing one term of matrix $A$
(see Eq.~(\ref{eq:PDE discretized})). Using linear combination of unitaries
(LCU) for each $A_k$ and adjoint $A^{\dagger}_k$, we construct
block-encodings $U_{A}$ for $A$ and $U_{A^\dagger}$ for $A^\dagger$. Next, we define:
\begin{eqnarray}
\begin{gathered}
L_1(\phi):=\ketbra{0}{0}\otimes U_A+\ketbra{1}{1}\otimes Ie^{i\phi}
\sim\left(\begin{array}{ccc}
A/\mathcal{N}_A&0&\dots
\\0&Ie^{i\phi}&\dots
\\\vdots&\vdots&\ddots
\end{array}\right);\\
L_2(\phi):=\ketbra{0}{0}\otimes U_{A^\dagger}+\ketbra{1}{1}\otimes Ie^{i\phi}
\sim\left(\begin{array}{ccc}
A^\dagger/\mathcal{N}_A&0&\dots
\\0&Ie^{i\phi}&\dots
\\\vdots&\vdots&\ddots
\end{array}\right),
\end{gathered}\label{eq: L definition AUX}
\end{eqnarray}
with normalization constant $\mathcal{N}_A$ which is composed from $\mathcal{N}_D\mathcal{N}_f\kappa$ according to $A=\sum_k A_k$ (see Thm.~\ref{theorem: 1 term robin}). The control-qubit states $\ket{0}$ and $\ket{1}$ determine whether we apply $U_A$ ($U_{A^\dagger}$) or global phase. The schematic matrix form in Eq.~(\ref{eq: L definition AUX})
assumes the explicit qubit ordering described as follows:
(i) First, ancillary qubits for measurement operations within unitary $U_A$;
(ii) Next, the single control qubit selecting between applying $U_A$
(or $U_{A^\dagger}$) and the global phase operation $I e^{i\phi}$;
(iii) Lastly, the main $n$-qubit register encoding the quantum state
associated with matrix $A$. This ordering ensures correct interpretation of
block-encoding matrix structures presented above.

Now we introduce two additional qubits for constructing another linear combination (LCU):
\begin{eqnarray}
\begin{gathered}
\frac{\mathcal{N}_A}{2}L_1(\pi)+\frac{\mathcal{N}_A}{2}L_2(-\pi)+
\frac{\mathcal{N}_B}{2}X\otimes\hat{O}_B
=\ketbra{0}{0}\otimes\frac{A+A^\dagger}{2}+X\otimes\frac{B}{2}
\sim\left(\begin{array}{ccc}
\frac{A+A^\dagger}{2}&\frac{B}{2}&\dots
\\\frac{B}{2}&0&\dots
\\\vdots&\vdots&\ddots
\end{array}\right)=U_{S_1};\\
i\frac{\mathcal{N}_A}{2}L_2(0)-i\frac{\mathcal{N}_A}{2}L_1(0)+
\frac{\mathcal{N}_B}{2}Y\otimes\hat{O}_B
=\ketbra{0}{0}\otimes \frac{A-A^\dagger}{2i}+Y\otimes\frac{B}{2}
\sim\left(\begin{array}{ccc}
\frac{A-A^\dagger}{2i}&\frac{B}{2i}&\dots
\\-\frac{B}{2i}&0&\dots
\\\vdots&\vdots&\ddots
\end{array}\right)=U_{S_2},
\end{gathered}\label{eq: block-encoding of S1 and S_2}
\end{eqnarray}
where
$X=\left(\begin{array}{cc}0&1\\1&0\end{array}\right)$,
$Y=\left(\begin{array}{cc}0&-i\\i&0\end{array}\right)$ are Pauli matrices;
$U_{S_1},U_{S_2}$ encode matrices $S_1,S_2$ (see Eq.~(\ref{eq: Hamiltonian main})).
The qubit ordering follows the same convention as previously described.
The details of oracle $\hat{O}_B$ and the normalization constant $\mathcal{N}_B$ are explicitly provided in Sec.~\ref{sub: block-encoding of B}.

Next, implement the $(\max(a,b),\lceil\log_2n_\xi\rceil,0)$-block-encoding
of the $2^{n_\xi}\times2^{n_\xi}$ matrix $\mathbf{x}_\xi$ (see Eq.~(\ref{eq: Schrodingerisation with time})) via
Thm.~\ref{Theorem: Amplitude-oracle for piece-wise polynomial function}, achieving the complexity:
\begin{enumerate}
\item $\mathcal{O}(n_\xi\log n_\xi)$ quantum gates;
\item $\lceil\log_2 n_\xi\rceil-1$ pure ancillas.
\end{enumerate}
Thus, the following theorem summarizes the results of this section.

\begin{theorem}[One-dimensional block-encoding]\label{theorem: block-encoding 1D}
    Let $H$ be a $2^{n+1} \cdot 2^{n_\xi} \times 2^{n+1} \cdot 2^{n_\xi}$ Hermitian matrix with sparsity $\kappa$. The matrix $H$ represents the Hamiltonian of an $n + n_\xi$-qubit system, governed by the equation:
    \[
        \frac{d\psi}{dt} = H\psi;
    \]
    where the Hamiltonian takes the form derived from the Schrodingerisation technique applied to the discretized PDE problem in the Eq.~(\ref{eq: PDE 1D MAIN}):
    \[
        H = S_1 \otimes \mathbf{x}_{\xi} + S_2 \otimes \mathbf{1}_{\xi}.
    \]
    The matrices \( S_1 \) and \( S_2 \) are given by:
    \[
        S_1 = \begin{pmatrix} \frac{A + A^\dagger}{2} & \frac{B}{2} \\ \frac{B}{2} & 0 \end{pmatrix}, \quad
        S_2 = \begin{pmatrix} \frac{A - A^\dagger}{2i} & \frac{B}{2i} \\ -\frac{B}{2i} & 0 \end{pmatrix},
    \]
    where
    \[
        A = \sum_{k=0}^{\eta-1} A_k, \quad 
        B = \sqrt{2^n} \, \text{diag}\{ v(x_0) + v'(x_0), \dots, v(x_{N-1}) + v'(x_{N-1}) \},
    \]
    where matrix elements of $2^n\times2^n$ matrix $A_k$ is given by
    \[
        (A_k\vec{u})_{ij} = \frac{f_k(x_i)}{(\Delta x)^{p_k}} \gamma_{ij}^{p_k},
    \]
    where \( \gamma_{ij}^{p_k} \) are the coefficients of the finite-difference scheme used ($p^k$ is the derivative degree), and \( \Delta x \sim 2^{-n} \) represents the grid size. Additionally, for the auxiliary $n_\xi$-qubit quantum register we have:
    \[
        \mathbf{1}_{\xi} = \text{diag}\{ 1, 1, \dots, 1 \}, \quad
        \mathbf{x}_{\xi} = \text{diag}\left( a_\xi, a_\xi + \frac{b_\xi - a_\xi}{2^{n_\xi} - 1}, \dots, b_\xi \right),
    \]
    where the domain \( [a_\xi, b_\xi] \) represents the auxiliary dimension introduced by the Schrodingerisation technique.

    We can then create a $(\mathcal{O}(\kappa ||H||_{\max}), \lceil \log_2 n_\xi \rceil + \lceil \log_2 n \rceil + \lceil \log_2 G \rceil + \lceil \log_2 \kappa \rceil + \lceil \log_2 \eta \rceil + 7, 0)$-block-encoding of $H$. The total resource cost scales as:
\begin{enumerate}
    \item $\mathcal{O}(\sum_{g=0}^{G_{v}}Q^{v}_gn\log n+\sum_{k=0}^{\eta-1}\left(\kappa_k n+\sum_{g=0}^{G_{f_k}}Q^{f_k}_gn\log n\right)+n_\xi\log n_\xi)$ quantum gates,
    \item $2n+2$ pure ancillas,
\end{enumerate}
where $\kappa_k$ is a sparsity of $A_k$; $G_{v}$ and $G_{f_k}$ are the number of pieces in $v(x)$, $f_k(x)$ respectively; $Q^{f_k}_g$, and $Q^{v}_g$ are the polynomial degrees, see Thm.~\ref{Theorem: Amplitude-oracle for piece-wise polynomial function}. We clarify the key notations in Table~\ref{tab:notations 1d}.

\end{theorem}
\begin{proof}
    This entire section provides a step-by-step explanation for constructing the block-encoding of $H$. Firstly, we create the block-encoding for $A_k$, see Thm.~\ref{theorem: block-encoding 1D}; then we combine them using the LCU method to form $A$. Simultaneously, we create block-encodings for diagonal matrices $B$ and $\mathbf{x}_\xi$ using Thm.~\ref{Theorem: Amplitude-oracle for piece-wise polynomial function}. Again using the LCU method, we combine a controlled version of the block-encoding of $A$ and $A^\dagger$, along with the block-encodings of $X \otimes B$ and $Y \otimes B$, to obtain the block-encodings of $S_1$ and $S_2$, denoted as $U_{S_1}$ and $U_{S_2}$, as shown in Eqs.~(\ref{eq: L definition AUX}), (\ref{eq: block-encoding of S1 and S_2}). Finally, using LCU one last time, we combine $U_{S_1}$, $U_{S_2}$, and the block-encoding of $\mathbf{x}_\xi$ to obtain the matrix $H$.
\end{proof}

\begin{remark}\label{remark: H_max}
    The constant factor \( ||H||_{\max} \) plays a crucial role in the overall algorithm for simulating the PDE. In Thm.~\ref{theorem: optimal block-hamiltonian simultaions}, we observe that the total gate count scales linearly with:
\[
\kappa ||H||_{\max} \sim \kappa \sum_{k=0}^{\eta-1} \mathcal{N}_{f_k} 2^{p_k n},
\]
where the \( 2^n \) factor arises from the discretized derivative degree $p_k$. Therefore, \( ||H||_{\max} \) is a critical factor that limits the quantum advantage in the number of partitions \( N = 2^n \) to a polynomial scale, as the best-known classical algorithm has a complexity of \( 2^{(p_k+1)n} \).

\end{remark}

\begin{table}[H]
    \centering
    \begin{tabular}{clc}
        \toprule
        \textbf{Symbol} & \textbf{Description} & \textbf{First Appearance} \\
        \midrule
        $n$ & Number of qubits representing the grid size & Eq.~(\ref{eq:PDE discretized}) \\
        $n_\xi$ & Number of qubits in the auxiliary register $\xi$ (Schrodingerisation) & Eq.~(\ref{eq: Schrodingerisation with time})\\
        $G$ & Max. number of pieces among all piece-wise continuous functions $v(x)$ and $\{f_k(x)\}$ & Thm.~\ref{Theorem: Amplitude-oracle for piece-wise polynomial function}\\
        $\kappa$ & Sparsity of the matrix $H$ & Sec.~\ref{subsec: sparsity}\\
        $\eta$ & Number of terms in the operator $\hat{A}$ & Eq.~(\ref{eq: PDE 1D MAIN})\\
        $||H||_{\max}$ & Largest absolute value among the elements of the matrix $H$ & Remark~\ref{remark: H_max}\\
        \bottomrule
    \end{tabular}
    \caption{List of key notations and their first appearances for $1$-D block-encoding implementation.}
    \label{tab:notations 1d}
\end{table}

Therefore, we achieve polynomial scaling with respect to the number of qubits, $n$ for the complexity of constructing the block-encoding, which establishes the efficiency of our method for the one-dimensional case. This construction leverages the structured sparsity of finite-difference matrices and the piecewise continuity of the physical functions to enable an efficient encoding of the Hamiltonian $H = S_1 \otimes x_\xi + S_2 \otimes \mathbf{1}_{\xi}$. The resulting block-encoding requires only logarithmic resources in resolution and linear resources in the number of physical and auxiliary dimensions, making it suitable for integration into quantum singular value transformation frameworks and further quantum simulation protocols.

\section{Generalization to Multi-dimensional cases}\label{sec: multidim}

In this section, we generalize the block-encoding construction presented  
in Sec.~\ref{sec: 1D} to a broader class of problems.  
Specifically, we derive the Hamiltonian simulation task for  
multi-dimensional linear PDEs with Robin boundaries~(\ref{eq: multidim boundary + initial})  
using the Schrodingerisation technique, and outline how to construct  
the corresponding block-encoding of the Hamiltonian. Additionally, we extend  
this approach to the time-dependent Hamiltonian case, as described  
in Eq.~(\ref{eq: Schrodingerisation with time}), demonstrating that it  
fits within the new multi-dimensional framework.

The general form of a multi-dimensional linear PDE that we consider  
is given below. It involves an operator $\hat{A}^{(d)}$ comprising  
a sum of space- and time-dependent differential terms, along with  
an inhomogeneous term $v(x_1,\dots,x_d,t)$. The system is subject to  
an initial condition and Robin boundary conditions in each spatial  
dimension:
\begin{eqnarray}
\begin{gathered}
\frac{\partial u(x_1,x_2,\dots,x_d,t)}{\partial t} =
\hat{A}^{(d)} u(x_1,\dots,x_d,t) + v(x_1,\dots,x_d,t); \qquad
\hat{A}^{(d)} = \sum_{k=0}^{\eta-1} \underbrace{f_k^{(d)}(x_1,\dots,x_d,t)
\frac{\partial^{p^{(1)}_k}}{\partial x_1^{p^{(1)}_k}} \cdots
\frac{\partial^{p^{(d)}_k}}{\partial x_d^{p^{(d)}_k}}}_{\hat{A}_k^{(d)}};\label{eq: multidim PDE first} \\
\begin{array}{rl}
\text{initial condition:} &
u(x_1, \dots, x_d, t=0) = u_0(x_1, \dots, x_d); \\[1ex]
\text{BC at } x_1 = a_1: &
\left. \frac{\partial u(x_1, \dots, x_d, t)}{\partial x_1} \right|_{x_1 = a_1}
+ \mathcal{A}_1^{(1)} u(x_1 = a_1,x_2, \dots,x_d, t)
= \mathcal{A}_2^{(1)}; \\
\text{BC at } x_1 = b_1: &
\left. \frac{\partial u(x_1, \dots, x_d, t)}{\partial x_1} \right|_{x_1 = b_1}
+ \mathcal{B}_1^{(1)} u(x_1 = b_1,x_2 \dots,x_d, t)
= \mathcal{B}_2^{(1)}; \\
\vdots & \\
\text{BC at } x_d = a_d: &
\left. \frac{\partial u(x_1, \dots, x_d, t)}{\partial x_d} \right|_{x_d = a_d}
+ \mathcal{A}_1^{(d)} u(x_1,\dots,x_{d-1},x_d = a_d, t)
= \mathcal{A}_2^{(d)}; \\
\text{BC at } x_d = b_d: &
\left. \frac{\partial u(x_1, \dots, x_d, t)}{\partial x_d} \right|_{x_d = b_d}
+ \mathcal{B}_1^{(d)} u(x_1,\dots,x_{d-1},x_d = b_d, t)
= \mathcal{B}_2^{(d)}. \\
\end{array}
\end{gathered}\label{eq: multidim boundary + initial}
\end{eqnarray}

Analogous to Sec.~\ref{section: discretization schrodingerisation and time dependence},  
we discretize all the spatial dimensions $x_i$; for simplicity,  
we assume that all dimensions have the same discretization size  
$N = 2^n$. Thus, we introduce an $n$-qubit register for each  
dimension $x_i$. Following the method~\cite{cao2023quantum}, as in  
Eq.~\ref{eq: get rid of time dependence}, we assign one more  
clock dimension $x_s$ (an $n_s$-qubit register) to eliminate  
all time dependence in $\hat{A}^{(d)}$ and $v$ by mapping  
$t \rightarrow x_s$. 

Thus, following the fundamental case described  
in Sec.~\ref{sec: 1D}, we require a block encoding of  
$A^{(d)}$, along with an amplitude oracle for $B^{(d)}$—a  
diagonal $2^{\sum n_i+n_s} \times 2^{\sum n_i+n_s}$ matrix with  
entries $Const \cdot v^\prime(x_1,\dots,x_d,t \rightarrow x_s)$  
on the diagonal—denoted by $\hat{O}_B^{(d)}$.

As before, we begin by constructing a block-encoding for the  
$1$-term of the matrix $A^{(d)}$:
\begin{eqnarray}
    A^{(d)}_k = f^{(d)}(\textbf{x}_1,\dots,\textbf{x}_d,\textbf{x}_s)\left[
    D_{p^{(1)}_k} \otimes \dots \otimes D_{p^{(d)}_k} \otimes  
    I^{\otimes n_s} \right], \label{eq: multidim A_k}
\end{eqnarray}
where $D_{p^{(i)}_k}$ denotes the $2^n \times 2^n$ discretized  
derivative operator $\frac{\partial^{p^{(i)}_k}}{\partial x_i^{p^{(i)}_k}}$.  
Each spatial dimension may in principle use the same or a distinct  
finite-difference scheme; for simplicity, we omit this detail  
in the notation. Common choices are illustrated in Table~\ref{table: famous schemes}.  
The matrix $f^{(d)}(\textbf{x}_1,\dots,\textbf{x}_d,\textbf{x}_s)$ is diagonal  
in the computational basis, with elements given by:
\[
\bra{i_1}^n \dots \bra{i_d}^n \bra{i_s}^{n_s} f^{(d)}(\textbf{x}_1,\dots,\textbf{x}_d,\textbf{x}_s)  
\ket{i_1}^n \dots \ket{i_d}^n \ket{i_s}^{n_s} =  
f^{(d)}(x_1 = a_1 + i_1 \Delta x, \dots, t = a_s + i_s \Delta x_s).  
\]

\subsection{Block-encoding of $B^{(d)}$ and Amplitude-oracle for $f^{(d)}(x_1, \dots, x_d, t)$}\label{subsec: multidim spatial block-encoding}

Firstly, we consider the diagonal unitary operator $\hat{O}_f^{(d)}$, which acts as  
\begin{eqnarray}
\begin{gathered}
    \hat{O}_f^{(d)} \ket{0}^{\lambda} \ket{i_1}^{n} \dots \ket{i_d}^{n} \ket{i_s}^{n_s} \\
    = \frac{1}{\mathcal{N}^{(d)}_f} f^{(d)}(x_1 = a_1 + i_1 \Delta x,  
    \dots, x_d = a_d + i_d \Delta x_d, t = a_s + i_s \Delta x_s)  
    \ket{0}^{\otimes \lambda} \ket{i_1}^{n} \dots \ket{i_d}^{n} \ket{i_s}^{n_s} + \dots  
\end{gathered}
\label{eq: oracle f matrix elements multidim}
\end{eqnarray}
The symbol $\dots$ in Eq.~(\ref{eq: oracle f matrix elements multidim})  
denotes terms where the first $\lambda$-qubit register is not  
in the zero state. Therefore, $\hat{O}_f^{(d)}$ is a  
$(\mathcal{N}_f^{(d)}, \lambda, 0)$-block-encoding that prepares  
the function $f^{(d)}_k$, analogous to Thm.~\ref{Theorem: Amplitude-oracle for piece-wise polynomial function}.  
Typically, the normalization constant satisfies  
$\mathcal{N}^{(d)}_f \sim \max_{x_1, \dots, x_d, t} \abs{f^{(d)}(x_1, \dots, x_d, t)}$.

The generalization of Thm.~\ref{Theorem: Amplitude-oracle for piece-wise polynomial function}  
is not straightforward, and currently no efficient general algorithm  
exists for multi-dimensional functions. Quantum signal processing (QSP),  
which is the key tool in Thm.~\ref{Theorem: Amplitude-oracle for piece-wise polynomial function},  
does not yet have a multi-dimensional extension (M-QSP) for arbitrary  
multivariate polynomials. The recent work~\cite{rossi2022multivariable}  
highlights the major challenges in this direction.  

Due to this limitation, we restrict our attention to a  
less general class of functions. Specifically, if the function is separable:
\begin{eqnarray}
    f^{(d)}(x_1, \dots, x_d, t) = g_1(x_1) \cdot g_2(x_2) \cdots  
    g_d(x_d) \cdot g_s(t),
\end{eqnarray}
where each $g_i$ is piece-wise continuous, then we can apply  
Thm.~\ref{Theorem: Amplitude-oracle for piece-wise polynomial function}  
to each dimension independently. In this case, the full oracle  
can be constructed as:
\begin{eqnarray}
    \hat{O}_f^{(d)} = \hat{O}_{g_1} \otimes \hat{O}_{g_2} \otimes \dots  
    \otimes \hat{O}_{g_d} \otimes \hat{O}_{g_s},  
    \label{eq: oracles separable}
\end{eqnarray}
where each $\hat{O}_{g_i}$ is assumed to be implemented using  
Thm.~\ref{Theorem: Amplitude-oracle for piece-wise polynomial function}.

Secondly, if the function is a superposition of separable functions,  
\begin{eqnarray}
    f^{(d)}(x_1, \dots, x_d, t) = \sum^{M-1}_{m=0} g_{m1}(x_1) \cdot g_{m2}(x_2)  
    \cdots g_{md}(x_d) \cdot g_{ms}(t),
    \label{eq: superposition of functions multidim}
\end{eqnarray}
then we can use the LCU construction based on Eq.~(\ref{eq: oracles separable})  
to implement the full oracle as:
\begin{eqnarray}
\begin{gathered}
    \hat{O}_f^{(d)} = \left[ H_W^{M} \otimes I^{\otimes nd + n_s} \right]  
    \left( \sum_{m = 0}^{M - 1} \ket{m}^{\lceil \log_2 M \rceil}  
    \bra{m}^{\lceil \log_2 M \rceil} \otimes \hat{O}_{g_{m1}} \otimes  
    \hat{O}_{g_{m2}} \otimes \dots \otimes \hat{O}_{g_{md}} \otimes  
    \hat{O}_{g_{ms}} + \dots \right) \\
    \cdot \left[ (H_W^{M})^T \otimes I^{\otimes nd + n_s} \right],  
\end{gathered}
\label{eq: multidimensional oracle superposition coordinate}
\end{eqnarray}
where $H_W^{M}$ is defined in Eq.~(\ref{eq: arbitrary sparcity}).  
The symbol $\dots$ denotes additional terms (e.g., identity operators)  
required to complete the sum into a unitary operator.

Finally, we also consider functions of the form  
\begin{eqnarray}
    f^{(d)}(x_1, \dots, x_d, t) = h\left( \sum^{M-1}_{m = 0}  
    g_{m1}(x_1) \cdot g_{m2}(x_2) \cdots g_{md}(x_d) \cdot g_{ms}(t)  
    \right), \label{eq: multidimensiona matreshka}
\end{eqnarray}
where $h(x)$ is a continuous real-valued function (which can  
be generalized to a complex piece-wise continuous one). For such  
functions, we employ the polynomial eigenvalue transformation (PET).  

\begin{theorem}[Polynomial eigenvalue transformation (Thm.~56 from~\cite{gilyen2019quantum})]  
\label{theorem: PET}  
Suppose $U$ is an $(\alpha, a, \varepsilon)$-encoding of a Hermitian matrix  
$A$. Let $\delta \geq 0$, and let $P_{\mathbb{R}} \in \mathbb{R}[x]$  
be a degree-$Q$ polynomial satisfying:  
\begin{itemize}
    \item For all $x \in [-1, 1]$: \quad $|P_{\mathbb{R}}(x)| \leq \frac{1}{2}$,
\end{itemize}
Then there exists a quantum circuit $\tilde{U}$ that is a  
$(1, a+2, 4Q\sqrt{\varepsilon/\alpha + \delta})$-encoding of  
$P_{\mathbb{R}}(A/\alpha)$. This circuit uses $Q$ applications  
of $U$ and $U^\dagger$, one controlled-$U$ gate, and  
$\mathcal{O}((a+1)Q)$ additional one- and two-qubit gates. Moreover,  
a classical description of the circuit can be computed in time  
$\mathcal{O}(\mathrm{poly}(Q, \log(1/\delta)))$.
\end{theorem}

Then, using PET from Thm.~\ref{theorem: PET} applied to  
the unitary in Eq.~(\ref{eq: multidimensional oracle superposition coordinate}),  
we can construct $\hat{O}_f^{(d)}$ corresponding to Eq.~(\ref{eq: multidimensiona matreshka}).  
Moreover, we can generalize to complex-valued $h(x)$ by applying  
PET separately to the real and imaginary parts and then combine them using LCU.
Additionally, we may use the indicator-function approach from  
Appendix~\ref{Appendix: indicator} to extend this construction  
to piece-wise continuous functions.  

\begin{remark} \label{remark: neural networks}  
The approach described in this subsection closely resembles a  
neural network (NN) structure~\cite{samek2021explaining}. First,  
we apply a superposition of 1-D linear dependencies via  
LCU (analogous to a linear layer), as shown in  
Eq.~(\ref{eq: multidimensional oracle superposition coordinate}).  
Then we apply the nonlinear transformation $h(x)$ using PET  
(analogous to an activation function), as in  
Eq.~(\ref{eq: multidimensiona matreshka}). By alternating  
these two steps, we can construct an amplitude oracle  
$\hat{O}_f^{(d)}$ for a wide class of functions.  
However, we emphasize that such constructions are costly to  
implement using quantum algorithms. Thus, we assume that the  
number of iterations (i.e., linear layers and activation  
functions) remains small in practical applications.  
\end{remark}  

Using the procedure described in this subsection, we construct  
$\hat{O}_f^{(d)}$ and $\hat{O}_{B}^{(d)}$, which serve as  
multi-dimensional analogues of $\hat{O}_f$ and $\hat{O}_B$  
from the previous Section.

The gate complexity of a separable oracle in  
Eq.~(\ref{eq: oracles separable}) scales linearly with $d$.  
The complexity of the linear superposition in  
Eq.~(\ref{eq: multidimensional oracle superposition coordinate}) is linear  
in $M$. The PET approach exhibits linear scaling with the  
polynomial degree $Q_{\text{PET}}$. Therefore, the overall complexity is  
given by:
\begin{eqnarray}
    \mathcal{O}\left(d \underbrace{(M Q_{\text{PET}})^{L_{\text{NN}}}}_{\text{NN-like approach}}  
    \underbrace{G Q n \log n}_{\text{1-D case}}\right) \text{ quantum gates,}  
    \label{eq: complexity multidimensional coordinate}
\end{eqnarray}
where $L_{\text{NN}}$ denotes the number of iterations (layers) as described in  
Remark~\ref{remark: neural networks}; $G$ is the number of pieces  
in the piecewise continuous function $g_i$, which are approximated  
by the polynomial of degree $Q$. We stress out that for a wide range of functions that has real world application $L_{\text{NN}}$ is a small number, see Table~\ref{table: amplitude-oracle-functions}.

\begin{table}[H]
\begin{small}
\centering
\caption{Selected examples of real-world functions $f(x_1,\dots,x_d,t)$ and required NN-like layers (LCU+PET pairs) for amplitude oracle construction. The symbol $\chi$ denotes an indicator function, e.g., $\chi_{[a,b]}(x) = 1$ if $x \in [a,b]$ and $0$ otherwise.}
\begin{tabular}{|l|l|c|p{4.2cm}|}
\hline
\textbf{Name / Use Case} & \textbf{Function $f(x_1,\dots,x_d,t)$} & \textbf{NN Layers} & \textbf{Application Domain} \\
\hline
Exponential source decay & $e^{-\gamma t} \cdot \chi_{[a,b]\times[c,d]}(x,y)$ & 1 & Localized heating and diffusion \\
\hline
Oscillatory driving force & $\sin(\omega_x x)\cos(\omega_y y)\sin(\omega_t t)$ & 1 & Electromagnetic and wave propagation \\
\hline
Gaussian heat source & $e^{-[(x - x_0)^2 + (y - y_0)^2]/\sigma^2} \cdot e^{-\lambda t}$ & 2 & Laser-induced heating, thermal modeling \\
\hline
Time-dependent diffusion coeff. & $D_0 (1 + \epsilon \sin(\omega t)) \cdot \chi_{[a,b]}(x)$ & 2 & Variable conductivity or diffusion \\
\hline
Two-mode squeezed vacuum & $\frac{1}{\sqrt{\pi} \cosh r} \cdot \exp\left(-\frac{1}{2}(x_1^2 + x_2^2) + \tanh r \cdot x_1 x_2\right)$ & 1 & Quantum optics, continuous-variable quantum states \\
\hline
Interacting field kernel & $\exp\left(-\lambda x_1^2 x_2^2\right)$ & 2 & Scalar quantum fields, $\phi^4$ theory \\
\hline
Radial potential well & $\chi_{[0,R]}(\sqrt{x^2 + y^2})$ & 1 & Radially symmetric traps or dots \\
\hline
Step function barrier & $V_0 \cdot \chi_{[x_1,x_2]}(x)$ & 1 & Potential wells, quantum tunneling problems \\
\hline
\end{tabular}
\label{table: amplitude-oracle-functions}
\end{small}
\end{table}

\subsection{Block-encoding for $d$-D Hamiltonian}

Let us now consider the rest part of the Eq.~(\ref{eq: multidim A_k})
\begin{eqnarray}
    D^{(d)}_{\{p_k\}}=D_{p^{(1)}_k} \otimes \dots \otimes D_{p^{(d)}_k} \otimes  
    I^{\otimes n_s}
    \label{eq: separable part multidim 1 term}
\end{eqnarray}
    
which corresponds to approximations of the derivatives in each dimension except clock dimension (eventually we add derivative in the clock dimension following the Eq.~(\ref{eq: get rid of time dependence})). We note that this structure has a separable view, hence the block encoding of $ D^{(d)}_{\{p_k\}}$ is a tensor product of block-encoding for each dimension. The individual $D_{p^{(i)}_k}$ was already build in the previous Section in Fig.~\ref{fig:1 term ROBIN} where we need to omit the spatial part $\hat{O}_f$; we provide an updated design in Fig.~\ref{fig: derivative circuit}.

\begin{figure}[H]
    \centering
    \includegraphics[width=1\linewidth]{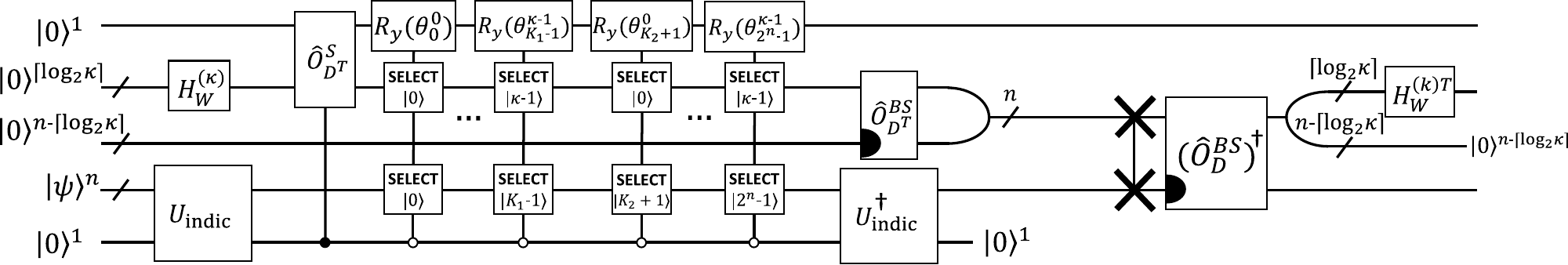}
    \caption{
    Circuit design for the block-encoding of the discretized operator
    $ D_{m}\sim \frac{\partial^m}{\partial x^m} $
    for the case of Robin boundaries (banded matrix $ D $ with some
    deviations near the boundaries). The circuit design repeating the logic in Fig.~\ref{fig:1 term ROBIN}.}
    \label{fig: derivative circuit}
\end{figure}

Now we have all the ingredients for implementing block-encoding of $A_k^{(d)}$
\begin{eqnarray}
    U_{A_k^{(d)}}=\hat{O}_f^{(d)}\cdot U_{D^{(d)}_{\{p_k\}}},
\end{eqnarray}
where “\(\cdot\)” denotes the multiplication of block-encodings~\cite{dalzell2023quantum}: the main \(n\)-qubit register is shared, while all other qubits are not.

Combining all the terms using LCU technique we suggest creating block-encoding of $A^{(d)}$. From this moment onward the multidimensional case repeat the logic of the previous Section precisely. The following Theorem concludes the Section.

\begin{theorem}[Multi-dimensional block-encoding]\label{theorem: block-encoding multidimensional}
Let $H$ be a $2^{nd+1} \cdot 2^{n_\xi} \cdot 2^{n_s} \times 2^{nd+1} 
\cdot 2^{n_\xi} \cdot 2^{n_s}$ Hermitian matrix with sparsity $\kappa$. 
The matrix $H$ represents the Hamiltonian of an $n + n_\xi +n_s$-qubit 
system, governed by the equation:
\[
\frac{d\psi}{dt} = H\psi;
\]
where the Hamiltonian takes the form derived from the Schrodingerisation 
technique applied to the PDE problem, see Eqs.~(\ref{eq: multidim PDE first}), 
(\ref{eq: get rid of time dependence}), (\ref{eq: Schrodingerisation with time}):
\[
H = \mathbf{1}  \otimes \textbf{p}_s\otimes \mathbf{1}_{\xi} 
+ S^{(d)}_1  \otimes \mathbf{x}_{\xi}
+ S^{(d)}_2 \otimes \mathbf{1}_{\xi}.
\]
The matrices \( S^{(d)}_1 \) and \( S^{(d)}_2 \) are given by:
\[
S^{(d)}_1 = \begin{pmatrix} \frac{A^{(d)} + {A^{(d)}}^\dagger}{2} & \frac{B^{(d)}}{2} \\ 
\frac{B^{(d)}}{2} & 0 \end{pmatrix}, \quad
S^{(d)}_2 = \begin{pmatrix} \frac{A^{(d)} - {A^{(d)}}^\dagger}{2i} & 
\frac{B^{(d)}}{2i} \\ -\frac{B^{(d)}}{2i} & 0 \end{pmatrix}
\]
Assuming that $f^{(d)}(x_1, \dots, x_d, t)$ and $v^{(d)}(x_1, \dots, x_d, t)$ 
have the form (\ref{eq: multidimensiona matreshka}) (and assuming $L_{\text{NN}}=1$), 
we can then create a 
$(\mathcal{O}(\kappa \|H\|_{\max}), d \lceil \log_2 n \rceil+\lceil \log_2 n_s \rceil + 
\lceil \log_2 n_\xi \rceil + \lceil \log_2 G \rceil + 
\lceil \log_2 \kappa \rceil + \lceil \log_2 \eta \rceil + 
\lceil \log_2 M \rceil + 4d + 5, 0)$-block-encoding of $H$. The total amount of resources scales as:
\begin{enumerate}
    \item $\mathcal{O}( M Q_{\text{PET}} G Q (d n \log n+n_s\log n_s) + n_\xi\log n_\xi+d \eta \kappa n)$ quantum gates,
    \item $\mathcal{O}(n)$ pure ancillas.
\end{enumerate}
The key notations are given in Table~\ref{tab:multidim-notations}.

\end{theorem}

\begin{proof}
The overall idea repeats the proof of Thm.~\ref{theorem: block-encoding 1D} 
with $A^{(d)}$ instead of $A$. The construction of $A^{(d)}$ is described 
above in this Section. The main difficulty in the multi-dimensional case 
is the block-encoding of $B^{(d)}$ and the Amplitude-oracle for 
$f^{(d)}(x_1, \dots, x_d, t)$, which is described in 
Sec.~\ref{subsec: multidim spatial block-encoding}.
\end{proof}

\begin{table}[H]
    \centering
    \begin{tabular}{clc}
        \toprule
        \textbf{Symbol} & \textbf{Description} & \textbf{First Appearance} \\
        \midrule
        $n$ & Number of qubits representing the grid size & Eq.~(\ref{eq: multidim A_k}) \\
        $d$ & Number of spatial dimensions in the original PDE & Eq.~(\ref{eq: multidim PDE first}) \\
        $n_\xi$ & Number of qubits in the auxiliary register $\xi$ (Schrodingerisation) & Eq.~(\ref{eq: Schrodingerisation with time}) \\
        $n_s$ & Number of qubits in the clock register & Eq.~(\ref{eq: get rid of time dependence})\\
        $G$ & Maximum number of pieces among all the piece-wise continuous functions $g_{mi}(x)$ & Eq.~(\ref{eq: superposition of functions multidim})\\
        $M$ & Polynomial degree in the activation function $h$ & Eq.~(\ref{eq: superposition of functions multidim})\\
        $Q_{\text{PET}}$ & Polynomial degree in the activation function $h$ & Eq.~(\ref{eq: multidimensiona matreshka})\\
        $\kappa$ & Sparsity of the matrix $H$ & Sec.~\ref{subsec: sparsity}\\
        $\eta$ & Number of terms in the operator $\hat{A}^{(d)}$ & Eq.~(\ref{eq: multidim PDE first}) \\
        $\|H\|_{\max}$ & Largest absolute value among the elements of the matrix $H$ & Remark~\ref{remark: H_max} \\
        \bottomrule
    \end{tabular}
    \caption{List of key notations and their first appearances used in the multi-dimensional setting.}
    \label{tab:multidim-notations}
\end{table}

Thus, we achieve a polynomial scaling in the number of qubits $n$ for the complexity of constructing the block-encoding, demonstrating the overall efficiency of our approach. This result confirms that the block-encoding construction remains tractable even in the multi-dimensional setting and validates the practicality of applying quantum signal transformation and linear combination of unitaries (LCU) techniques for high-dimensional PDE Hamiltonians.

\section{Quantum method for numerical simulating of PDEs}\label{section: QSVT}

In this Section we consider the direct application of block-encoding construction from the Thms.~\ref{theorem: block-encoding 1D},\ref{theorem: block-encoding multidimensional} for numerical simulating of PDEs. The idea is to use the quantum singular value transformation from the Thm.~\ref{theorem: optimal block-hamiltonian simultaions} achieving the evolution operator $H\rightarrow e^{iHt}$.

\begin{theorem}[Optimal block-Hamiltonian simulation(Theorem 58 from \cite{gilyen2019quantum})]\label{theorem: optimal block-hamiltonian simultaions}
    Let $t\in(0, \infty)$, $\epsilon\in (0,1)$ and let $U_H$ be an $(\alpha,a,0)$-block-encoding of the $2^n\times2^n$ Hamiltonian H. Then we can implement an $\epsilon$-precise Hamiltonian simulation unitary $V$ which is a $(2,a+2,\epsilon)$-block-encoding of $e^{iHt}$, with resources no greater than
    \begin{enumerate}
        \item $\Omega (t,\epsilon)$ uses of $U_H$ or its inverse;
        \item $1$ use of controlled-U or its inverse;
        \item $\Omega(\alpha t,\epsilon)(16a+50)+4$ one-qubit gates;
        \item $\Omega(\alpha t,\epsilon)(12a+38)$ CNOT gates;
        \item $\mathcal{O}(n)$ ancilla qubits.
    \end{enumerate}
    $\Omega(\alpha t,\epsilon)$ is implicitly defined through inequality for the truncation parameter $g\geq \Omega(\alpha t,\epsilon)$
    
    \[ \frac{1.07}{\sqrt{g}}\left(\frac{\alpha et}{2g}\right)^g\leq\epsilon. \]
    The inequality determines scaling of the algorithm (Lemma $59$ from \cite{gilyen2019quantum})
    \[ \Omega(\alpha t,\epsilon)=\mathcal{O}\left(\alpha t+\frac{\ln(1/\epsilon)}{\ln(e+\frac{\ln(1/\epsilon)}{\alpha t})}\right) \]
\end{theorem}

The following two Theorems implement the evolution part from Fig.~\ref{fig:Schrodingerisation} for one-dimensional and two dimensional cases.

\begin{theorem}[Evolution operator for one-dimensional case]\label{theorem: evolution 1 D}
Let $t\in(0, \infty)$, $\epsilon\in (0,1)$ and suppose we have a spatial discretized PDE (\ref{eq: PDE 1D MAIN}) as in (\ref{eq:PDE discretized}):
\begin{eqnarray}
\begin{gathered}
    \frac{\partial \vec{u}}{\partial t} = A\vec{u} + \vec{v}; \quad  \vec{u} = \left( u_0 \;\; u_1 \;\; \cdots \;\; u_{N-1} \right)^\top;\qquad \vec{v}=\left(v_0+v^\prime_0,v_1+v^\prime_1,\dots,v_{N-1}+v^\prime_{N-1}\right)^\top;\\
    A=\sum_kA_k;\qquad \left(A_k\vec{u}\right)_i= \frac{f_k(x_i,t)}{(\Delta x)^{p_k}}\sum_{m=i-s^L_{p_k}}^{i+s^R_{p_k}}\gamma_{p_km} u_m\approx \left. f_k(x,t)\frac{\partial^{p_k}u(x,t)}{\partial x^{p_k}} \right|_{x = x_i},
    \end{gathered}
\end{eqnarray}

where the discretization parameter $N=2^n$ ($n$ is the number of qubits). Then we can implement the $(2, \lceil \log_2 n \rceil +\lceil \log_2 n_\xi \rceil + \lceil \log_2 G \rceil + \lceil \log_2 \kappa \rceil + \lceil \log_2 \eta \rceil + 9,\epsilon)$-block-encoding of the evolution operator $e^{-iHt}$ as in Fig.~\ref{fig:Schrodingerisation} with resources not greater than
\begin{enumerate}
    \item $\mathcal{O}(\left(\kappa ||H||_{\max}t+\frac{\ln(1/\epsilon)}{\ln(e+\frac{\ln(1/\epsilon)}{\kappa ||H||_{\max} t})}\right)\cdot\left(\sum_{g=0}^{G_{v}}Q^{v}_gn\log n+\sum_{k=0}^{\eta-1}\left(\kappa_k n+\sum_{g=0}^{G_{f_k}}Q^{f_k}_gn\log n\right)+n_\xi\log n_\xi\right)$ quantum gates,
    \item $2n+3$ pure ancillas,
\end{enumerate}
where $\kappa_k$ is a sparsity of $A_k$; $G_{v}$ and $G_{f_k}$ are the number of pieces in $v(x)$, $f_k(x)$ respectively; $Q^{f_k}_g$, and $Q^{v}_g$ are the polynomial degrees, see Thm.~\ref{Theorem: Amplitude-oracle for piece-wise polynomial function}. We clarify the key notations in Table~\ref{tab:notations 1d}.
\end{theorem}
\begin{proof}
    Applying Optimal block-Hamiltonian simulation (Thm.~\ref{theorem: optimal block-hamiltonian simultaions}) on the block-encoding implemented as in Thm.~\ref{theorem: block-encoding 1D}.
\end{proof}

\begin{theorem}[Evolution operator for multidimensional PDE]\label{theorem: evolution multidim}
Let $t\in(0, \infty)$, $\epsilon\in (0,1)$ and suppose we have a $d$-dimensional spatial discretized PDE (\ref{eq: multidim PDE first}) with the function of interest encoded as $2^{nd}$ vector $\vec{u}$ ($n$-qubits for each dimension):
\begin{eqnarray}
    \begin{gathered}
        \frac{\partial \vec{u}(t)}{\partial t} =
        A^{(d)} \vec{u}(t) + \vec{v}(t); \qquad
        A^{(d)}=\sum_{k=0}^{\eta-1}A^{(d)}_k; \\ 
        A^{(d)}_k=f^{(d)}(\textbf{x}_1,\dots,\textbf{x}_d,\textbf{x}_s)\left[
    D_{p^{(1)}_k} \otimes \dots \otimes D_{p^{(d)}_k} \otimes  
    I^{\otimes n_s} \right],
    \end{gathered}
\end{eqnarray}
where $D_{p^{(i)}_k}$ denotes the $2^n \times 2^n$ discretized  
derivative operator $\frac{\partial^{p^{(i)}_k}}{\partial x_i^{p^{(i)}_k}}$. Then we can implement the $(2, d \lceil \log_2 n \rceil + 
\lceil \log_2 n_s \rceil+\lceil \log_2 n_\xi \rceil + \lceil \log_2 G \rceil + 
\lceil \log_2 \kappa \rceil + \lceil \log_2 \eta \rceil + 
\lceil \log_2 M \rceil + 4d + 7,\epsilon)$-block-encoding of the evolution operator $e^{-iHt}$ as in Fig.~\ref{fig:Schrodingerisation} with resources not greater than
\begin{enumerate}
    \item $\mathcal{O}(\left(\kappa ||H||_{\max}t+\frac{\ln(1/\epsilon)}{\ln(e+\frac{\ln(1/\epsilon)}{\kappa ||H||_{\max} t})}\right)\cdot\left(\eta M Q_{\text{PET}} G Q (d n \log n+n_s\log n_s) + 
           n_\xi \log n_\xi + d \eta \kappa n\right)$ quantum gates,
    \item $\mathcal{O}(n)$ pure ancillas,
\end{enumerate}
We clarify the key notations in Table~\ref{tab:multidim-notations}.

\begin{corollary}[Postselection]
    Consequently, if we have a black-box quantum access to the initial condition $\vec{u}(t=0)$, then we can get the quantum state $\vec{u}(t)/||\vec{u}(t)||$ with additional resources  for $QFT^{-1}$: $\mathcal{O}(n_\xi^2)$ and probability of success $p_{\text{success}}\sim ||\vec{u}(t)||^2/||\vec{u}(0)||^2$ for measuring $P>0$, see \cite{analog} and Fig.~\ref{fig:Schrodingerisation}.
\end{corollary}

\end{theorem}

Despite the efficiency of our block-encoding construction for multi-dimensional PDEs, it is important to recognize the limitations in the scaling of the overall quantum algorithm. As discussed in Remark~\ref{remark: H_max}, the presence of a large block-encoding normalization constant proportional to $\kappa \|H\|_{\max}$ increases the total circuit depth. Consequently, the exponential quantum advantage in the number of grid points $N = 2^n$ is lost. The quantum algorithm exhibits a total gate complexity of $\mathcal{O}(2^{np} n \log n)$, where $p = \max_k \sum_{i=1}^{d} p_k^{(i)}$ is the highest total (accumulated through all the dimensions $d$) derivative order across all terms $k$ of the PDE.

In contrast, the best classical algorithms scale as $\mathcal{O}(2^{np + dn})$ \cite{leveque2007finite}, where $d$ is the spatial dimension. This implies that the quantum algorithm still achieves a polynomial advantage in terms of grid size $n$, with a speedup of approximately $2^{nd}$ — which is substantial in high-dimensional regimes but not exponential with respect to $n$.

\begin{remark}
While the quantum algorithm efficiently simulates the PDE evolution as a quantum state, it does not directly yield a classical solution in the form of function values at specific grid points. Instead, the solution is represented as a quantum state whose amplitudes encode the discretized solution vector. Extracting detailed classical information requires post-processing through measurement, which can be limited in resolution due to quantum sampling constraints \cite{nielsen2002quantum}. Nevertheless, this quantum-accessible representation enables efficient estimation of global properties \cite{10.1145/3188745.3188802,huang2020predicting} (e.g., norms, expectation values, spectral features) that are often computationally expensive to obtain classically. Thus, the algorithm should be viewed not as a replacement for classical solvers, but as a powerful alternative for problems where compact encoding and global observables are of primary interest.
\end{remark}

Finally, we highlight the exponential improvement in scaling with respect to the number of dimensions $d$. As shown in Thm.~\ref{theorem: evolution multidim}, the quantum algorithm scales linearly with the number of dimensions $\mathcal{O}(d)$. This effectively removes the \emph{curse of dimensionality}, which is the dominant bottleneck for classical solvers of high-dimensional PDEs. Hence, the quantum approach has exponential advantage in $d$, and polynomial advantage in $n$.

\section{Numerical tests}\label{section: numericals}
In this section, we follow the procedure outlined in Sec.~(\ref{sec: 1D}) to construct a quantum circuit for the Hamiltonian simulation of a one-dimensional heat equation subject to Robin boundary conditions, as Eq.~(\ref{eq:heat equation}) shows. To efficiently simulate this quantum circuit on a classical computer, we developed a quantum circuit simulator capable of executing quantum circuit instructions composed of one-qubit rotations and multi-controlled one-qubit rotations (see Appendix~\ref{appendix: multi-control}). In this numerical experiment, we explicitly construct an "instruction" containing the quantum circuit gate information, and subsequently utilize our quantum circuit simulator to perform the time evolution, where each computational step is explicitly implemented through elementary one-qubit rotations. It is noteworthy that the multi-controlled versions of one-qubit rotations are not further decomposed into elementary one-qubit rotations and CNOT gates in this numerical experiment, as our simulator directly executes these multi-controlled gates. This simplification, however, does not impact the dominant complexity of our algorithm. Furthermore, we employed the quantum signal processing (QSP) phase-factors solver from \cite{dong2021efficient} to determine the QSP phase factors in the Hamiltonian simulation. We adapted their Matlab package \footnote{ \url{https://qsppack.gitbook.io/qsppack}} into Python to better suit our computational framework. Further details can be found on our GitHub repository, see \cite{your_repo_2025}. Additionally, state preparation is not considered here; instead, the initial state vector is directly provided by the classical computer.

The specific form of the heat equation under investigation is given by Eq.~(\ref{eq:heat equation}),
\begin{equation}\label{eq:heat equation}
\begin{cases}
\frac{\partial u(x,t)}{\partial t} = \frac{\partial^2 u(x,t)}{\partial x^2}, \qquad (t,x) \in [0,T] \times [0,L_x], \\
\frac{\partial u}{\partial x}\Big|_{x=0} + \mathcal{A}_1 u(a,t) = \mathcal{A}_2, \\
\frac{\partial u}{\partial x}\Big|_{x=L_x} + \mathcal{B}_1 u(b,t) = \mathcal{B}_2, \\
u(0,x) = \sin\left(\frac{\pi (L_x-x)}{2L_x}\right), \qquad x \in [0, L_x],
\end{cases}
\end{equation}
with coefficients chosen as
\begin{equation*}
\mathcal{A}_1 = \frac{1}{2}, \qquad \mathcal{B}_1 = 1, \qquad \mathcal{A}_2 = \mathcal{B}_2 = \frac{1}{4}, \qquad L_x = 10.
\end{equation*}
The spatial domain $[0,L_x]$ is discretized uniformly with step size $\Delta x$, defining grid points $x_i = a + i \Delta x$. We set the number of grid points to $N_x=2^{n_x}$, with $n_x$ representing the number of qubits allocated to the spatial variable, and $\Delta x = L_x/(N_x-1)$. Employing a five-point central difference scheme augmented with the ghost points technique at the boundaries transforms the problem into an inhomogeneous system of ordinary differential equations:
\[
\frac{d\vec{u}}{dt} = A\vec{u} + \vec{v},
\]
where the matrix $A$ is defined in Eq.~(\ref{Eq.matirxA}), and the vector $\vec{v}$ is expressed as
\begin{equation*}
\vec{v} = \text{diag}\left(-\frac{7\mathcal{A}_2}{3\Delta x}, \frac{\mathcal{A}_2}{6\Delta x}, 0, \cdots, 0, -\frac{\mathcal{B}_2}{6\Delta x}, \frac{7\mathcal{B}_2}{3\Delta x}\right) = \text{diag}\left(-\frac{7}{12}, \frac{1}{24}, 0, \cdots, 0, -\frac{1}{24}, \frac{7}{12}\right).
\end{equation*}

Given that $\vec{v}$ is time-independent, rather than invoking the transformation approach from Eq.~(\ref{Eq.homogeneous transformation}), we reformulate the system equivalently into an augmented homogeneous form:
\begin{equation}
\label{eq.heat homogeneous}
\frac{d}{dt}
\underbrace{
\left[
\begin{array}{c}
\vec{u} \\
\vec{v}
\end{array}
\right]}_{\vec{w}}
= 
\underbrace{\left[
\begin{array}{cc}
A & I \\
0 & 0
\end{array}
\right]}_S
\left[
\begin{array}{c}
\vec{u} \\
\vec{v}
\end{array}
\right] ; \qquad
\left[
\begin{array}{c}
\vec{u}(0) \\
\vec{v}(0)
\end{array}
\right]
 = 
\left[
\begin{array}{c}
\vec{u}_0 \\
\vec{v}
\end{array}
\right]
\end{equation}
where $I$ denotes the identity matrix. This modification slightly reduces computational complexity since the block-encoding of $B$ (cf. Eq.~(\ref{Eq.homogeneous transformation})) no longer participates in the quantum singular value transformation, but is invoked only once for the initial state preparation—an aspect not considered here. However, the principal contributing term in the overall complexity remains unaffected.





Since the ODE system (\ref{eq.heat homogeneous}) exhibits non-Hermitian dynamics, we adopt the Schrödingerisation method \cite{jin2024quantum} by introducing an auxiliary variable $\xi$, whose computational domain is set as $[-L_\xi, L_\xi]$. Analogously, this domain is discretized uniformly with a spatial step size $\Delta \xi$, resulting in grid points $\xi_i = -L_\xi + i\Delta\xi$. We denote the number of grid points by $N_\xi=2^{n_\xi}$, where $n_\xi$ indicates the number of qubits allocated to this auxiliary spatial variable, and thus $\Delta \xi = 2L_\xi/(N_\xi-1)$. Using Schrödingerisation, we transform the original ODE system into the following Schrödinger-type equation:
\begin{equation}
\frac{d \vec{a}}{dt} = \left(S_1 \otimes \mathbf{x}_{\xi} + S_2 \otimes \mathbf{1}_{\xi}  \right) \vec{a},
\end{equation}
where $\vec{a} \in \mathbb{C}^{2^{n_x+1+n_\xi}}$,
\[
S_1 = (S+S^\dagger)/2, \quad S_2 = (S-S^\dagger)/2i, \quad \mathbf{x}_\xi = \text{diag}\left(-L_\xi, -L_\xi + \Delta\xi, \cdots, L_\xi-\Delta\xi, L_\xi \right)
\]
and $\mathbf{1}_\xi\in \mathbb{C}^{N_\xi\times N_\xi}$ is an identity matrix.

In our numerical experiments, we set $n_x = 5$, $n_\xi = 8$ and $L_\xi = 12$, resulting in a total of $n_{\text{total}} = 28$ qubits for the simulation of the heat equation. Given the absence of an analytic solution for this equation, we employ the classical Forward Euler method to numerically approximate the solution, which serves as a benchmark for evaluating the accuracy of our quantum simulation. The comparison between the classical numerical solution and our quantum simulation results is illustrated in Fig.~\ref{fig:comparison_quantum_classical}

This numerical experiment evaluates Thm.~\ref{theorem: evolution 1 D} by applying it to the one-dimensional heat conduction equation with Robin boundary conditions. The corresponding simulation results at various evolution times are shown in Fig.~\ref{fig:comparison_quantum_classical}, demonstrating the robustness of the theorem in practice. Compared with classical methods, the quantum simulations exhibit small errors, and the fidelity of the resulting quantum states consistently exceeds 0.99999. These results suggest that the proposed algorithm may offer a viable approach to solving linear partial differential equations on quantum hardware.

\begin{figure}[H]
    \centering
    \begin{subfigure}{\textwidth}
        \centering
        \begin{subfigure}{0.48\textwidth}
            \centering
            \includegraphics[width=\linewidth]{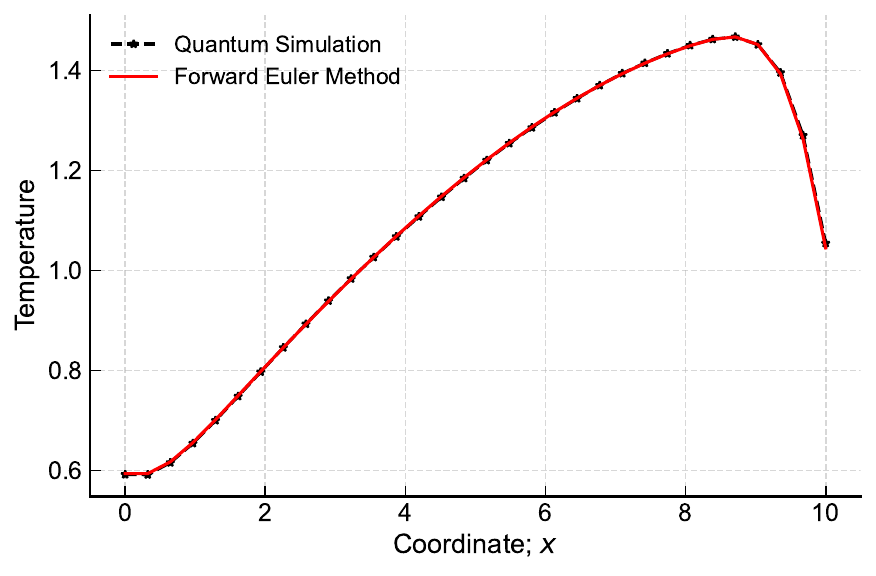}
        \end{subfigure}
        \hfill
        \begin{subfigure}{0.48\textwidth}
            \centering
            \includegraphics[width=\linewidth]{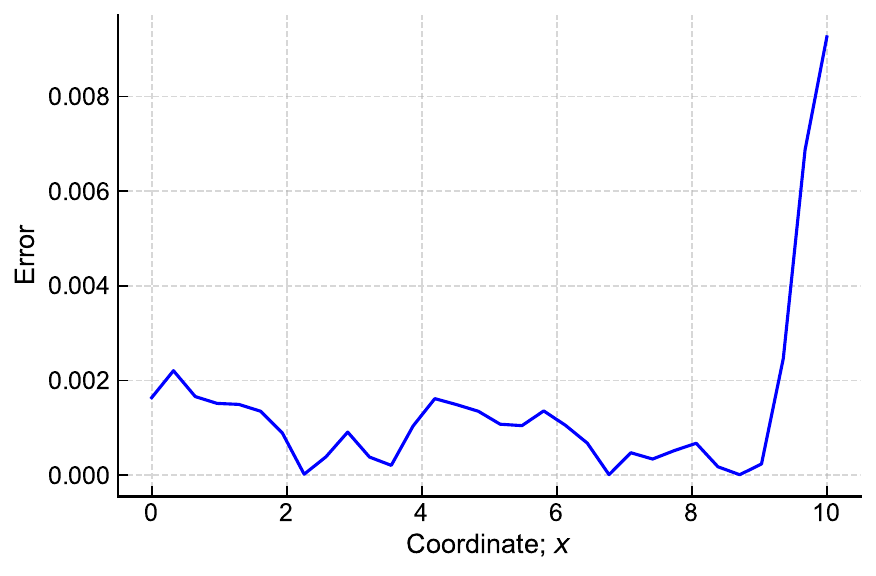}
        \end{subfigure}
        \caption{Simulation time $T=0.2$, MSE $=5.1961025\times 10^{-5}$, Fidelity $=0.99999326$, Gate count = $2.3 \times 10^6$.}
        \vspace{0.3cm}
    \end{subfigure}
    \begin{subfigure}{\textwidth}
        \centering
        \begin{subfigure}{0.48\textwidth}
            \centering
            \includegraphics[width=\linewidth]{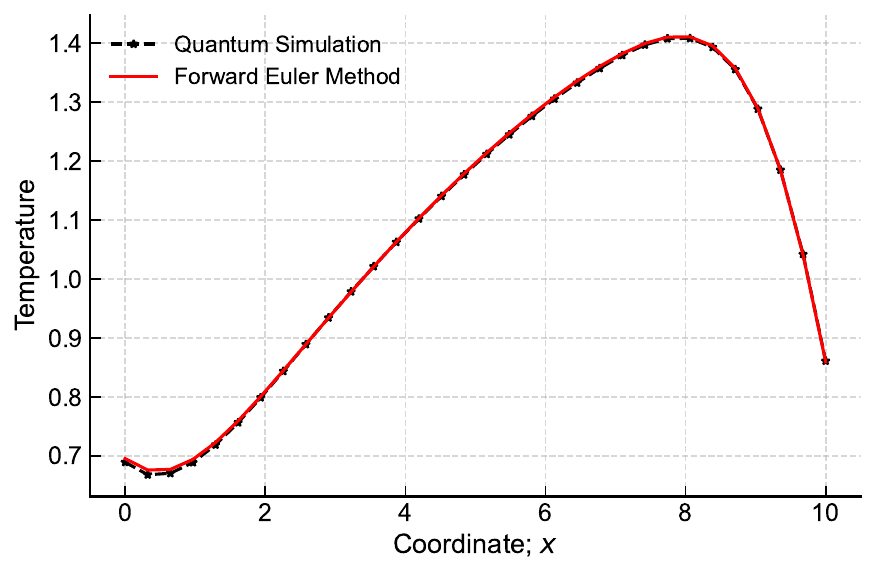}
        \end{subfigure}
        \hfill
        \begin{subfigure}{0.48\textwidth}
            \centering
            \includegraphics[width=\linewidth]{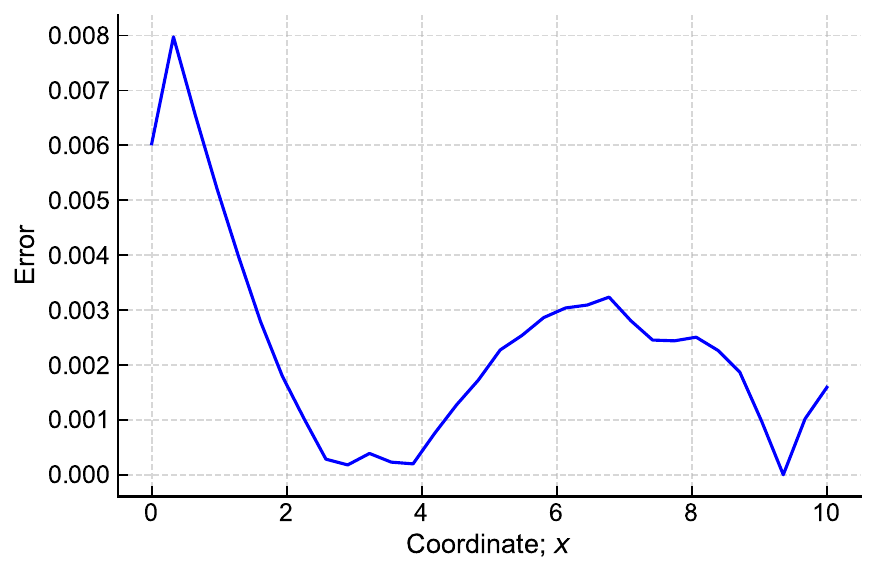}
        \end{subfigure}
        \caption{Simulation time $T=0.6$, MES $=9.1157\times 10^{-6}$, Fidelity $= 0.99999573$, Gate count = $6.9 \times 10^6$.}
        \vspace{0.3cm}
    \end{subfigure}
    \begin{subfigure}{\textwidth}
        \centering
        \begin{subfigure}{0.48\textwidth}
            \centering
            \includegraphics[width=\linewidth]{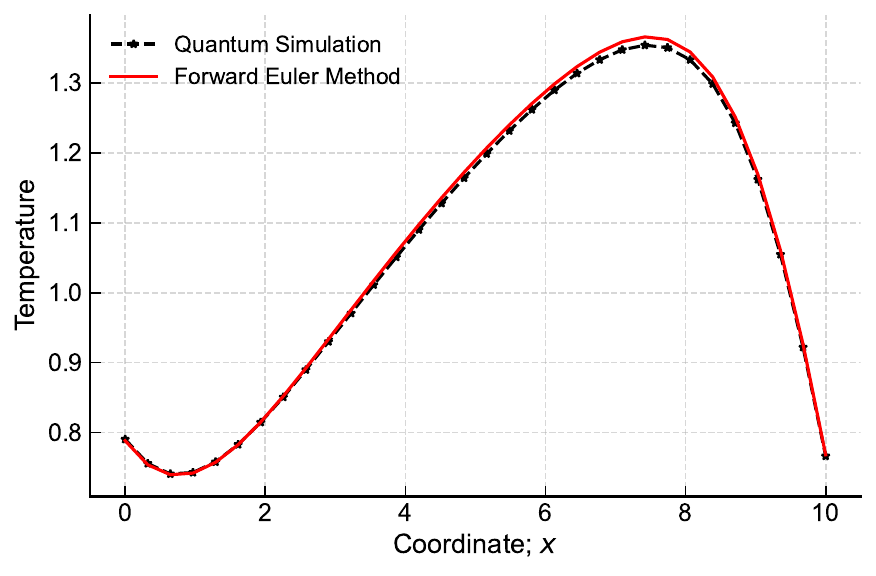}
        \end{subfigure}
        \hfill
        \begin{subfigure}{0.48\textwidth}
            \centering
            \includegraphics[width=\linewidth]{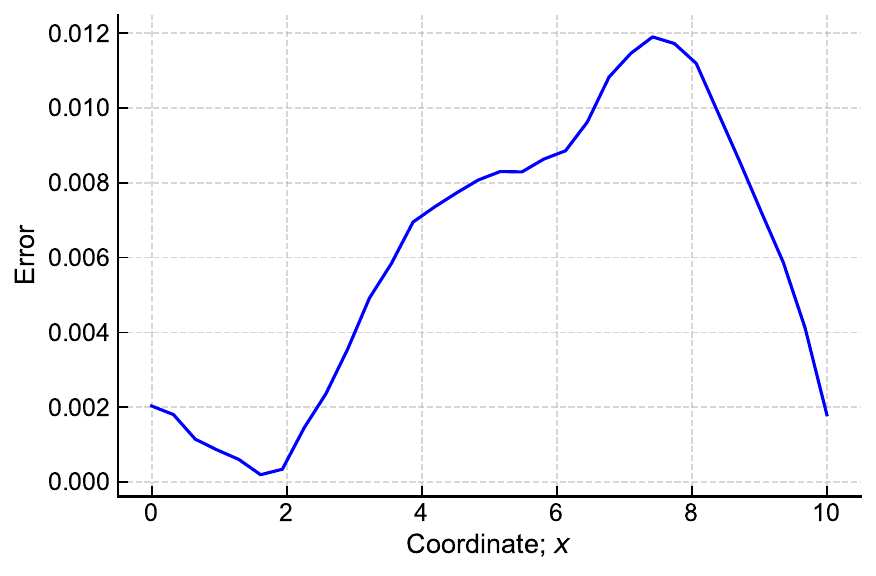}
        \end{subfigure}
        \caption{Simulation time $T=1.0$, MES $=5.0903\times 10^{-5}$, Fidelity$=0.99999139$, Gate count = $11.5 \times 10^6$.}
    \end{subfigure}

    \caption{Comparison of numerical solutions for Eq. (\ref{eq:heat equation}) obtained using the classical Forward Euler method and quantum Hamiltonian simulation (Thm.~\ref{theorem: evolution 1 D}) at different evolution times, with corresponding error distributions. Each subfigure displays: (left) the numerical results at the specified simulation time, and (right) the error between classical and quantum solutions. The simulation time, mean error squared (MES) and quantum state fidelity between two solutions are shown in each subfigure's caption. Here MES is defined as $\sum_{i = 0}^{N_x-1}(u_{\text{quantum},i} - u_{\text{Euler}, i})^2/N_x$. $u_{\text{quantum},i}$ and $u_{\text{Euler}, i}$ represent the quantum simulation's and Forward Euler method's numerical result at grid point $x_i$ respectively. Gate count 
    represents the number of multi-controlled one-qubit rotations, see 
    Def.~\ref{def:multiconrol operator} and Fig.~\ref{fig:Multi_control_qc}.}
    \label{fig:comparison_quantum_classical}
\end{figure}

\section{Discussion and Future Work}\label{section: conclusion}

In this work, we have developed an explicit, oracle-free quantum
framework for numerically simulating a broad class of linear PDEs
with variable coefficients and general inhomogeneous terms, supporting
Robin, Neumann, Dirichlet, and periodic boundary conditions. Our
approach generalizes previous methods by providing direct, efficient
constructions of amplitude block-encodings for discretized PDE Hamiltonians,
enabling quantum simulation of cases that have not been addressed
in the prior literature. Unlike many existing algorithms, which rely
on abstract oracles or require costly quantum arithmetic, our method
offers explicit gate-level instructions for block-encoding construction
and practical Hamiltonian simulation using quantum signal processing.
This achieves polynomial speedup in grid resolution and exponential
speedup in spatial dimension $d$ over classical finite-difference
schemes.

For future work, we aim to extend our framework to more
sophisticated boundary conditions, such as those where the coefficients
$\mathcal{A}_1, \mathcal{A}_2, \mathcal{B}_1, \mathcal{B}_2$ can depend on
time and spatial coordinates, or may have even more complex structure
than those considered in Eq.~(\ref{eq: PDE 1D MAIN}). For example, interface BC, artificial BC, transparent BC, dynamic BC, etc. Other directions
include exploring quantum preconditioning techniques to reduce
$\|H\|_{\max}$ and thereby improve the quantum signal processing
complexity bounds established in our main Thms.~\ref{theorem: evolution 1 D},\ref{theorem: evolution multidim}. Additional research
could address improved quantum state preparation, measurement strategies,
and possible extensions to nonlinear PDEs and the impact of noise
on overall circuit performance.

\section*{Acknowledgments}
N. Liu acknowledges funding from the Science and Technology Commission of Shanghai Municipality (STCSM) grant no. 24LZ1401200 (21JC1402900), by NSFC grants No.12471411 and No. 12341104, the Shanghai Jiao Tong University 2030 Initiative, the Shanghai Science and Technology Innovation Action Plan (24LZ1401200) and the Fundamental Research Funds for the Central Universities. N. Guseynov acknowledges funding from NSFC grant W2442002.

\bibliography{references}

\appendix

\section{The complexity analysis for classical methods} \label{appendix: complexity analysis for classical method}
In this section, we analyze the computational complexity for classical methods, including explicit schemes such as the forward Euler method, and implicit schemes such as Crank–Nicolson and the backward Euler method \cite{gautschi2011numerical, leveque2007finite}. We count floating-point operations (FLOPs) as the complexity unit. These methods discretize time and $d$-dimensional space to iteratively evolve the solution via matrix updates. Let $N = 2^n$ denote the number of grid points per dimension, so the total discretized domain contains $N^d$ points. Denote $x_j, j = 0,1,\cdots,N^d-1$ as the grid points. Next, let $\Delta t$ represent the time step, and $t_j=j\Delta t, j=0,1,\cdots$.

Firstly, we analyze the complexity of the Forward Euler method.  After discretising the spatial variables, the linear PDEs $\partial u/\partial t = \mathcal{L}u$ \cite{evans2022partial}, where $\mathcal{L}$ is a derivative operator with respect to the spatial variables and $m$ denotes the highest-order derivative in $\mathcal{L}$, are transformed into an ODE system $d\vec{u}/dt = A^{(d)}\vec{u} + \vec{b}$, where $A^{(d)} \in \mathbb{R}^{N^d \times N^d}$ is a sparse matrix with sparsity $\mathcal{O}(m)$. Here, $\vec{u}(t) = [u(t, x_0),\cdots, u(t, x_{N^d - 1})]^T \in \mathbb{R}^{N^d}$ represents the numerical solution at each grid point at time $t$. The time variable is then discretised as follows:
\begin{equation*}
\vec{u}_{n+1} = (\Delta tA^{(d)}+I)\vec{u}_n + \Delta t \vec{b}_n,
\end{equation*}
where $\vec{u}_n = \vec{u}(t_n)$. To update the solution at the next time step, each iteration requires $\mathcal{O}(N^d m)$ multiplications. However, since the explicit method is conditionally stable, the time step $\Delta t$ and spatial step $\Delta x$ must satisfy the CFL condition \cite{leveque2007finite}
\[
\Delta t \leq \alpha (\Delta x)^m, 
\]
where $\alpha$ is a constant. Consequently, to evolve the solution to the terminal time $T$, the total number of time steps required is
\[
\frac{T}{\Delta t} \geq \frac{T}{\alpha (\Delta x)^m}  = \mathcal{O}(TN^m).
\]
Thus, the overall complexity is $\mathcal{O}(TN^{d+m}m)$.

Next, we consider the implicit method. After discretising the time variable, we must solve a linear system at each time step:
\[
A^{(d)}_n \vec{u}_n = \vec{b}_n,
\]
where the matrix $A^{(d)}_n \in \mathbb{R}^{N^d \times N^d}$ and the vector $\vec{b}_n \in \mathbb{R}^{N^d}$ are known. To solve this linear system, we may employ a direct solver (e.g., Gaussian elimination\cite{kress2012numerical}, etc.) or an iterative solver (e.g., conjugate gradient, multigrid methods \cite{kress2012numerical, saad2003iterative, leveque2007finite}, etc.). In any case, the complexity of solving the linear system at each time step is at least $\mathcal{O}(N^d)$. 

Although implicit methods are unconditionally stable and there are no constraints on the time step $\Delta t$, the total error for the numerical PDE solution is given by \cite{leveque2007finite}
\begin{equation}
    \epsilon \sim \mathcal{O}(\Delta t) + \mathcal{O}(\Delta x^g),
\end{equation}
where $g$  represents the maximal cumulative accuracy order of the finite-difference derivative approximation, as detailed in Table~\ref{table: famous schemes}. To ensure that the temporal discretisation error does not dominate the overall error scaling, it must satisfy $\Delta t \sim (\Delta x)^g$. Therefore, simulating the total time $T$ requires $T/\Delta t \sim T N^g$ time steps. Consequently, the total complexity is $\mathcal{O}(TN^{d+g})$. 

In summary, the total complexity is $\mathcal{O}(TN^{d+m}m)$ for the explicit method, and $\mathcal{O}(TN^{d+g})$ for the implicit method.

\section{Important definitions}\label{appendix: important definitions}

In this Appendix, we introduce essential mathematical and quantum computing
concepts used throughout the paper, along with an indexing scheme.

\begin{definition}[Multi-control operator]
\label{def:multiconrol operator}
Let $U$ be an $m$-qubit unitary and $b$ a binary string of length $n$.
We define $C_U^b$ as a quantum operation acting on $n + m$ qubits,
which applies $U$ to the last $m$ qubits only when the first $n$ qubits
are in state $\ket{b}$. In Appendix~\ref{appendix: multi-control}, we show
that this operator can be constructed using $C^1_U$ gates, along with
$16n - 16$ single-qubit gates and $12n - 12$ C-NOT gates.

\[
C^b_U = \ket{b}\bra{b} \otimes U + \sum\limits_{i \neq b} \ket{i}\bra{i} \otimes I^{\otimes m}
\]

\end{definition}

\begin{definition}[Pure ancilla]
\label{def:pure ancilla}
Let $U$ be a unitary on $n + m$ qubits. If for any input $\ket{\psi}$,
\[ U\ket{0}^{m}\ket{\psi} = \ket{0}^{m}\ket{\phi}, \]
where $\ket{\phi}$ is a quantum state independent of the ancillas,
then we call the $m$-qubit register a pure ancilla for $U$.
These qubits begin and end in the $\ket{0}$ state and can be reused.
In Fig.~\ref{fig:Multi_control_qc}(c), we exploit $m-2$ ancillas
for constructing controlled operations efficiently. Their use is split
into (i) an entanglement phase to facilitate computation, and (ii) a
disentanglement phase to reset them.
\end{definition}

\begin{remark}
Pure ancillas are temporary helper qubits that do not need
measurement. Since they return to $\ket{0}$, they may be
recycled without additional reset.
\end{remark}

\begin{definition}[Block-encoding (adapted from~\cite{gilyen2019quantum})]\label{def:block-encoding}
Let $A$ be an $n$-qubit operator. Given $\alpha, \epsilon \in \mathbb{R}_+$
and integers $a, s \in \mathbb{Z}_+$, we say that a unitary $U$ acting on
$a + s + n$ qubits is an $(\alpha, s, \epsilon)$-block-encoding of $A$ if:
\[ U \ket{0}^a \ket{0}^s \ket{\psi} = \ket{0}^a \ket{\phi}^{n+s}, \]
with:
\[ (\bra{0}^s \otimes I) \ket{\phi}^{n+s} = \tilde{A} \ket{\psi},
\quad \text{and} \quad \| A - \alpha \tilde{A} \| \leq \epsilon. \]
The unitary $U$ may use pure ancillas internally, as explained above.
A visual overview of this scheme is shown in Fig.~\ref{fig:block_encoding_general}.
\end{definition}

\begin{figure}[h]
\centering
\includegraphics[width=0.4\textwidth]{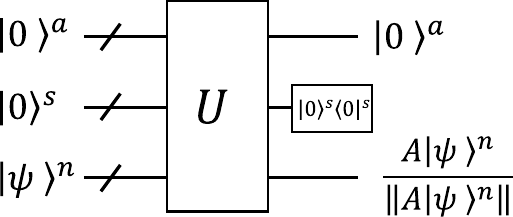}
\caption{General layout of block-encoding a matrix $A$.
The measurement of the second register in $\ket{0}$ projects the
main register onto the desired state.}
\label{fig:block_encoding_general}
\end{figure}

\section{Explicit quantum circuits for the multi-control unitaries}
\label{appendix: multi-control}

\begin{figure}[H]
\centering
\subcaptionbox{}{
    \includegraphics[width=0.283\textwidth]{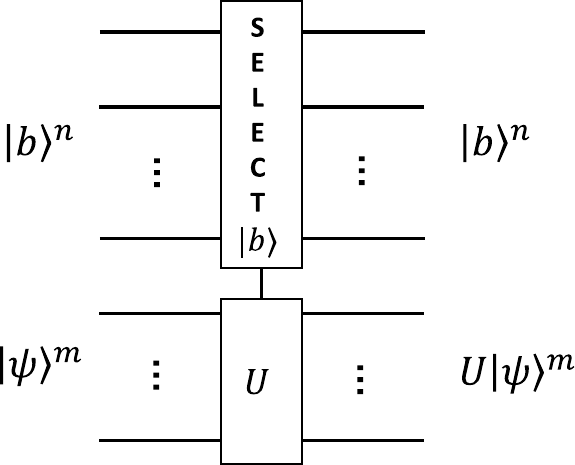}}
\hspace{0.1\textwidth}
\subcaptionbox{}{
    \includegraphics[width=0.48\textwidth]{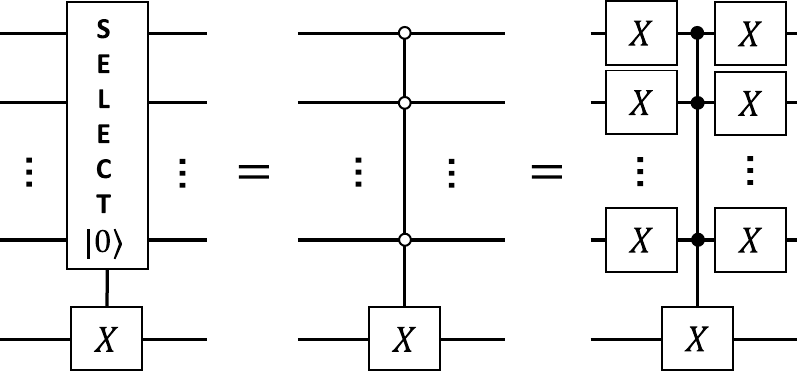}}
\vspace{2.0em}
\subcaptionbox{}{
    \includegraphics[width=0.68\textwidth]{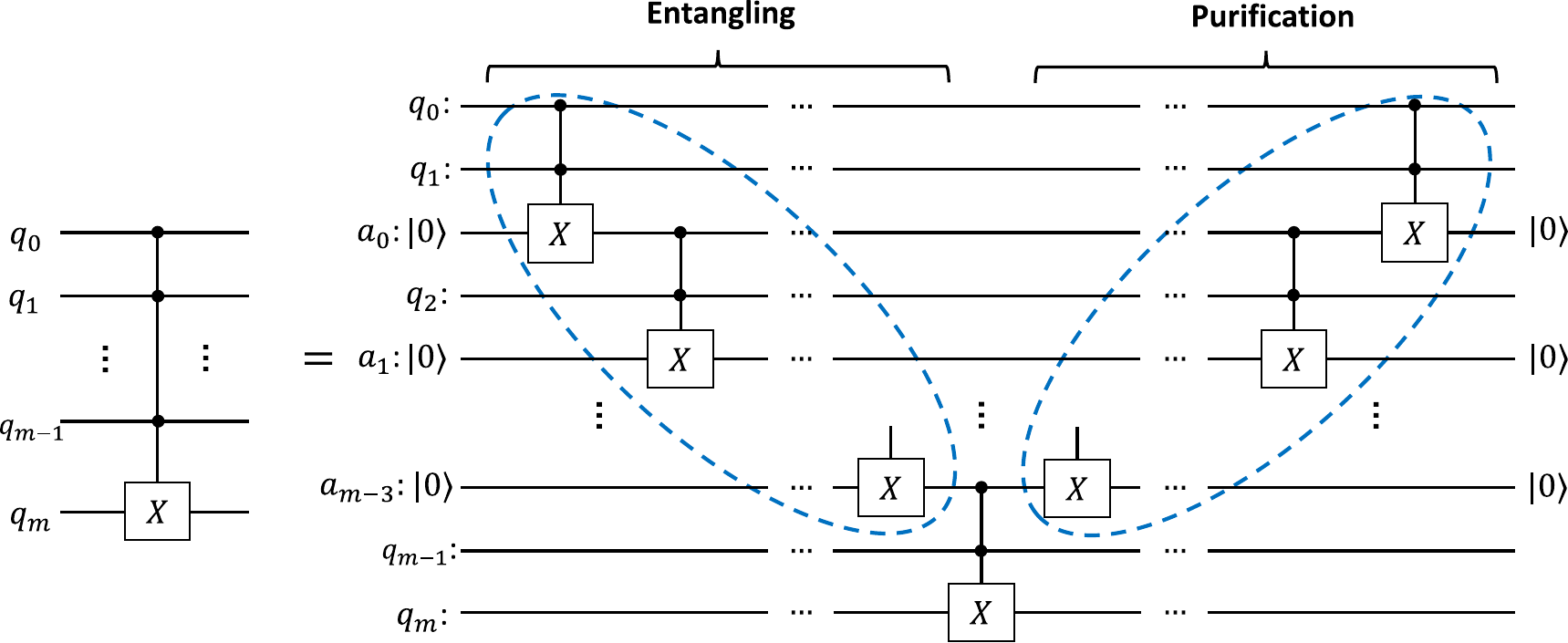}}
\caption{(a) Structure of the general multi-control gate $C_U^b$
from Definition~\ref{def:multiconrol operator}. The select box
indicates that $U$ is applied only when the control register is
$\ket{b}^n$. (b) Quantum circuit for $C_X^{00\dots 0}$. (c) Circuit
for $C_X^{11\dots 1}$, used throughout this work. It uses $2m - 3$
Toffoli gates and $m - 2$ pure ancillas.}
\label{fig:Multi_control_qc}
\end{figure}

This section presents our approach for implementing multi-controlled gates
and computing their resource requirements, as shown in
Fig.~\ref{fig:Multi_control_qc}. It is known~\cite{nielsen2002quantum}
that a Toffoli gate can be decomposed into 6 C-NOT and 8 single-qubit gates.
Therefore, a multi-control unitary $C_U^b$ can be constructed using one
$C^1_U$, along with $16n - 16$ single-qubit and $12n - 12$ C-NOT gates.

We assume all multi-control gates are of the form $C_X^{11\dots 1}$ for
simplicity. Additionally, the construction uses $n - 1$ auxiliary qubits
initialized to $\ket{0}$—these are the pure ancillas.

To estimate the gate complexity of a controlled unitary $C^1_U$, we
replace each C-NOT in $U$ with a Toffoli gate, which requires 6 C-NOT
and 8 single-qubit gates. Single-qubit gates are replaced with their
controlled versions, which can be implemented using 2 C-NOT and 2
single-qubit gates. This follows standard decomposition strategies
from~\cite{nielsen2002quantum}.

\section{Indicator for Register State Comparison with Desired Threshold}\label{Appendix: indicator}

We introduce an ancillary qubit that serves as an indicator,  
determining whether \( k \geq K_n \), where \( k \) is the  
index of the state in the computational basis and \( K_n =  
\lceil (K - a) / \Delta x \rceil \) is the discrete coordinate  
of \( K \). The initialization of this \textbf{indicator qubit}  
is illustrated in Fig.~\ref{fig:K_greater_or_not} and explained in Fig.~\ref{fig:OR}. This unitary operation,  
denoted \( U^K_c \), has a complexity that does not exceed  
the following:

\begin{enumerate}
    \item \( 16n + 34 \) single-qubit operations,
    \item \( 12n - 4 \) CNOT gates,
    \item \( n - 1 \) pure ancillas.
\end{enumerate}

\begin{figure}[H]
    \centering
    \includegraphics[width=0.5\textwidth]{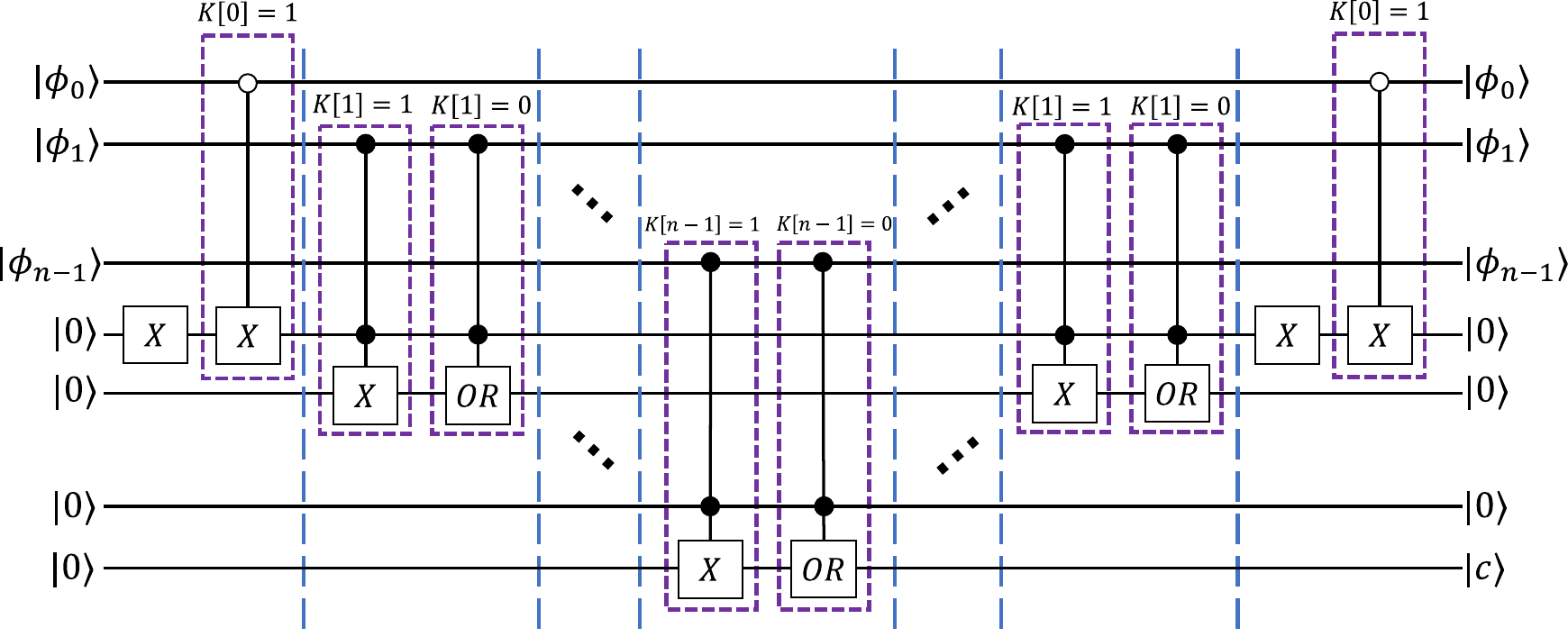}
    \caption{Circuit design for initializing the indicator qubit.  
    This circuit functions similarly to a classical comparator,  
    comparing the computational basis state in the upper register  
    with \( K \), and setting the last qubit to \( \ket{c} = \ket{1} \)  
    if \( \phi \geq K_n \). The circuit utilizes \( n-1 \) ancillary  
    qubits, resetting them to the \( \ket{0} \) state. A classical  
    array \( K[i] \) holds the binary representation of \( K_n \),  
    with the value of \( K[i] \) dictating the applied gate.}
    \label{fig:K_greater_or_not}
\end{figure}

\begin{figure}[H]
    \centering
    \includegraphics[width=0.4\textwidth]{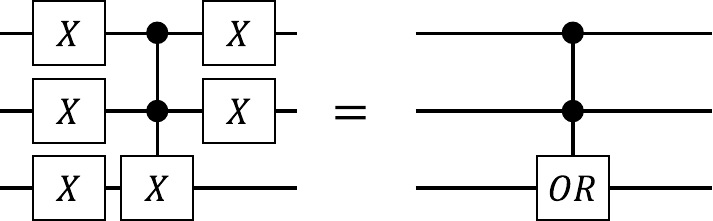}
    \caption{Representation of the OR gate used in the circuit  
    shown in Fig.~\ref{fig:K_greater_or_not}.}
    \label{fig:OR}
\end{figure}

\end{document}